\documentclass[11pt,a4paper]{article}
\usepackage{etex}

\usepackage{pictex}
\usepackage{mathrsfs}
\usepackage{amssymb}
\usepackage{amsfonts}
\usepackage{amsmath}
\usepackage{xcolor}
\usepackage{tikz}
\usepackage{caption}
\usetikzlibrary{shapes,snakes}

\title{On the packing chromatic number of subcubic outerplanar graphs}

\author{Nicolas Gastineau$^{1}$,  P\v remysl Holub$^{2}$ and Olivier Togni$^{3}$}

\date{\today}

\newtheorem{remark}{Remark}

\newcounter{mathitem}
\newenvironment{mathitem}
  {\begin{list}{{$(\roman{mathitem})$}}{
   \setcounter{mathitem}{0}
   \usecounter{mathitem}
   \setlength{\topsep}{0pt plus 2pt minus 0pt}
   \setlength{\parskip}{0pt plus 2pt minus 0pt}
   \setlength{\partopsep}{0pt plus 2pt minus 0pt}
   \setlength{\parsep}{0pt plus 2pt minus 0pt}
   \setlength{\leftmargin}{20pt}
   \setlength{\itemsep}{0pt plus 2pt minus 0pt}}}
  {\end{list}}

\newcounter{prostredi}
\def\theprostredi{\arabic{prostredi}}

\newenvironment{theorem}{\par\bigskip\noindent%
\refstepcounter{prostredi}{\bf Theorem \theprostredi.}\quad\bgroup\sl }
{\egroup\par\bigskip\endtrivlist}%
\newenvironment{proposition}{\par\bigskip\noindent%
\refstepcounter{prostredi}{\bf Proposition \theprostredi.}\quad\bgroup\sl }
{\egroup\par\bigskip\endtrivlist}%
\newenvironment{lemma}{\par\bigskip\noindent%
\refstepcounter{prostredi}{\bf Lemma \theprostredi.}\quad\bgroup\sl }
{\egroup\par\bigskip\endtrivlist}%
\def\vejde#1{\unskip
\nobreak\hfill\penalty50\hskip1em\hbox{}\nobreak\hfill
\hbox{#1}}

\newenvironment{proof}{\par
\noindent%
{\bf Proof.}\quad\bgroup}%\sl }
{\egroup\vejde{\rule{2.5mm}{2.5mm}}\par\bigskip\endtrivlist}%
{\egroup\vejde{\rule{2.5mm}{2.5mm}}\par\bigskip\endtrivlist}%
\newenvironment{observation}{\par\bigskip\noindent%
\refstepcounter{prostredi}{\bf Observation
\theprostredi.}\quad\bgroup\sl }
{\egroup\par\bigskip\endtrivlist}%
\newcounter{prostralph}
\def\theprostralph{\Alph{prostralph}}

\newenvironment{propositionA}[1]{\par\bigskip\noindent%
\refstepcounter{prostralph}{\bf Proposition \theprostralph{} {#1}.}\quad\bgroup\sl }
{\egroup\par\bigskip\endtrivlist}%
\newcounter{prostrclaim}
\def\theprostrclaim{\arabic{prostrclaim}}
{\egroup\vejde{$\square$}\par\bigskip\endtrivlist}%
{\egroup\vejde{$\square$}\par\bigskip\endtrivlist}%

\newcommand{\noi}{\noindent}

\begin{document}

\maketitle

\footnotetext[1]{LAMSADE UMR7243,
\textit{PSL, Univ. Paris-Dauphine, France}; 
e-mail: {\tt nicolas.gastineau@u-bourgogne.fr}}

\footnotetext[2]{Department of Mathematics, University of
West Bohemia;
European Centre of Excellence NTIS -
New Technologies for the Information Society;
P.O. Box 314, 306 14 
Pilsen, Czech Republic; e-mail: {\tt holubpre@kma.zcu.cz}}

\footnotetext[3]{Le2I FRE2005,
\textit{Universit\'e Bourgogne Franche-Comt\'{e}, F-21000 Dijon, France}; \\
e-mail: {\tt olivier.togni@u-bourgogne.fr}}

\begin{abstract}
\noi 

%The question of whether subcubic graphs have finite packing chromatic number is not open anymore since almost all cubic graphs of girth at least $2k+2$ has packing chromatic number at least $k$ (for every integer $k$). 
Although it has recently been proved that the packing chromatic number is unbounded on the class of subcubic graphs, there exists subclasses in which the packing chromatic number is finite (and small). These subclasses include subcubic trees, base-3 Sierpiński graphs and hexagonal lattices.
In this paper we are interested in the packing chromatic number of subcubic outerplanar graphs. We provide asymptotic bounds depending on structural properties of the outerplanar graphs and determine sharper bounds for some classes of subcubic outerplanar graphs.

%However, that not the cases for some subclasses, including subcubic trees, base-3 Sierpiński graphs and hexagonal lattices.
%In this paper, we prove that is also not the case for some subcubic outerplanar graphs. We provide asymptotic bounds depending on structural properties of the weak dual of the outerplanar graphs and determine sharper bounds on some classes of subcubic outerplanar graphs.

\medskip
\noindent {\bf Keywords:} packing colouring, packing chromatic number, outerplanar graphs, subcubic graphs.
\smallskip

\noindent {\bf AMS Subject Classification: 05C12, 05C15, 05C70} 

\end{abstract}

\section{Introduction}

Throughout this paper, we consider undirected simple graphs only, and for definitions and notations not defined here we refer to \cite{bondy_murty}.

Let $G$ be a graph and $c$ a vertex $k$-colouring of $G$, i.e., a mapping $c: V(G) \rightarrow \{1,2,\dots, k\}$. We say that $c$ is a {\em packing $k$-colouring} of $G$ if vertices coloured with the same colour $i$ have pairwise distance greater than $i$. The {\em packing chromatic number} of $G$, denoted by $\chi_{\rho}(G)$ is the smallest integer $k$ such that $G$ has a packing $k$-colouring; if there is no such integer $k$ then we set $
\chi_{\rho}(G)= \infty$. For a class  of graphs $\mathcal{C}$, we say that the packing chromatic number of $\mathcal{C}$ is finite if there exists a positive integer $k$ such that $\chi_{\rho}(G)\leq k$ for every graph $G\in \mathcal{C}$.

%Let $\mathbb{Z}^{2}$ denote the planar square lattice, $\mathscr{T}$ denote the planar triangular lattice and $\mathscr{H}$ denote the planar hexagonal lattice.

The concept of a packing colouring of a graph, introduced by Goddard et al. in \cite{God} under the name broadcast colouring, is inspired by frequency planning in wireless systems, in which it emphasizes the fact that signals can have different powers, providing a model for the frequency assignment problem.
The packing chromatic number of lattices has been studied by several authors: for the infinite square lattice $\mathbb{Z}^{2}$, Soukal and Holub in \cite{SO2010} proved that
$\chi_{\rho}(\mathbb{Z}^{2})\le 17$, while Ekstein et al. in \cite{EK2010} showed that $12\le\chi_{\rho}(\mathbb{Z}^{2})$. Recently, Martin et al. in \cite{Martsquare} improve the bounds by showing that $13\le \chi_{\rho}(\mathbb{Z}^{2}) \le 15$. For the infinite hexagonal grid $\mathscr{H}$, Fiala et al. in \cite{FI2009} showed that $\chi_{\rho}(\mathscr{H})\le7$, Kor{\v{z}}e and Vesel in \cite{AV2007} proved that $\chi_{\rho}(\mathscr{H})\ge7$. Finbow and Rall in \cite{Fin} proved that the infinite triangular grid $\mathscr{T}$ is not packing colourable, i.e., $\chi_{\rho}(\mathscr{T})=\infty$. 
Moreover, Kor{\v{z}}e and Vesel in \cite{AV2018} proved that the infinite octagonal lattice has packing chromatic number $7$.
The packing chromatic number of the Cartesian product of some graphs was investigated in \cite{BrePa,FI2009,jonck}. 
Also, the packing chromatic number has been studied for further graph classes in \cite {BrePa,BresKlav,God,HolubD,TO2010}.
The computational complexity has been also studied: determining whether a graph has packing chromatic number at most $4$ is an NP-complete problem \cite{God} and determining whether a tree has packing chromatic number at most $k$ (with a tree and $k$ on input) is also an NP-complete problem \cite{FiCo}.

Sloper in \cite{Slo} showed that the infinite complete ternary tree $T$ has $\chi_{\rho}(T)=\infty$ while any subcubic tree $T$ is packing $7$-colourable, hence it is natural to ask if all graphs with maximum degree $3$ (often so-called {\em subcubic graphs}) have finite packing chromatic number. This question was raised by Goddard et al.~\cite{God}. Recently, a second open question has been proposed about the packing chromatic number of $S(G)$, when $G$ is subcubic \cite{BresKlav3,Gasto} ($S(G)$ being the graph obtained from $G$ by subdividing each edge once).
Recently, Balogh, Kostochka and Liu \cite{Bal17} proved that, for any integer $k$, almost all cubic graphs of order $n$ and of girth at least $2k+2$ have packing chromatic number greater than $k$, hence answering negatively the question of Goddard et al.  Moreover, an explicit construction of an infinite family of subcubic graphs with unbounded packing chromatic number have been found very recently \cite{BreFe}.
%However, explicit constructions known so far are only for subcubic graphs with packing chromatic number up to 14 \cite{BresKlav2,Gasto}. 
Some subclasses of subcubic graphs were also under consideration, see e.g. \cite{BrePa,BresKlav}.

Outerplanar graphs form a class of structured graphs (containing the class of trees), which are generally easy to colour. Our aim is to find some  classes of subcubic outerplanar graphs, which have finite packing chromatic number. We define these classes by giving restrictions on their structure (number of faces of different types), or, equivalently, on their weak dual. Note that, when a graph is not connected, we can colour each component separately satisfying the distance constraints of a packing colouring and the resulting colouring is packing as well. Thus, throughout the rest of this paper, we will consider connected outerplanar graphs only.

The paper is organized as follows. Section 2 presents an upper bound for $2$-connected subcubic outerplanar graphs without internal face, i.e., for which the weak dual is a path. Then, in Section 3, we use results from Section 2 in order to determine asymptotic bounds for some larger classes of subcubic outerplanar graphs restricted by the number of (internal) faces. In Section 4, we improve bounds from Section 3 for some specific classes of subcubic outerplanar graphs with a specific structure. Finally, in the last section, we present lower bounds for the packing chromatic number of subcubic outerplanar graphs and give concluding remarks. Table \ref{table1} summarizes the main results of this paper. 

\begin{table}[ht]
\begin{center}
\begin{tabular}{|c|c|c|}
  \hline
 Condition on the subcubic outerplanar graph $G$ & $\ell$ & Section \\ \hline
 $G$ is $2$-connected with no internal face & $15$ & 2  \\ \hline
 $G$ is $2$-connected with at most $k$ internal faces & $17 \times 6^{3k}-2$ & 3 \\ \hline
 $G$ is connected with at most $k'$ faces & $ 9 \times 6^{k'}-2$ & 3 \\ \hline
 $G$ is $2$-connected with one internal face & 51 & 4 \\ \hline
 $G$ is connected with no internal face and with the block graph a path & 305 & 4 \\ \hline
\end{tabular}
\end{center}
\caption{Classes of subcubic outerplanar graphs and values of $\ell$ for which every relevant graph $G$ satisfies $\chi_{\rho}(G)\le \ell$.}
\label{table1}
\end{table}

\subsection{Preliminaries}

Let $G$ be a graph and $A\subset V(G)$. We denote $G-A$ the subgraph of $G$ after deletion of all vertices of $A$ from $G$ and all edges incident to some vertex of $A$ in $G$. We further denote $G[A]$ the subgraph of $G$ induced by $A$, or equivalently, $G[A]=G-(V(G)\setminus A)$. Specifically, for $x\in V(G)$, $G-x$ denotes the subgraph of $G$ after deletion of $x$ and all edges incident to $x$ from $G$. 

An {\em outerplanar graph} $G$ is a planar graph such that there exists a planar drawing of $G$ for which all vertices belong to the outer face. When it is $2$-connected, it can be represented by a {\em boundary cycle} $C$ containing all vertices of $G$, with non-crossing chords dividing the interior of $C$ into {\em faces}. A face $F$ of $G$ is called an {\em internal face} if $F$ contains more than two chords of $G$, and an {\em end face} of $G$ if $F$ contains only one chord of $G$; note that all remaining edges of an end face belong to $C$.

The {\em weak dual} of $G$, denoted by $\mathcal{T}_G$, is the graph with the set of all faces of $G$, except the outer face, as vertex set, and the edge set $E(G)=\left\{ FF' \vert \, F \mbox{ and }F' \mbox{ have an edge in common}\right\}$. We denote by $u_F$ the vertex of $\mathcal{T}_G$ corresponding to the face $F$ of $G$ and sometimes we identify a face $F$ and the corresponding vertex $u_F$ of $\mathcal{T}_G$. It is well known that the weak dual of a connected outerplanar graph is a forest and of a $2$-connected outerplanar graph is a tree.
Note that an end face of an outerplanar graph $G$ corresponds to a leaf of $\mathcal{T}_G$ and that an internal face of $G$ corresponds to a vertex of degree at least $3$ in $\mathcal{T}_G$.
Obviously, every end face of a $2$-connected outerplanar graph contains at least one vertex of degree 2.

For a graph $G$, the {\em block graph} of $G$, denoted by $\mathcal{B}_G$, is the graph where vertices of $\mathcal{B}_G$ represent all maximal 2-connected subgraphs of $G$ (usually called {\em blocks}) and two vertices of $\mathcal{B}_G$ are adjacent whenever the corresponding blocks share a cut vertex.

For any $G_1\subset G$, let $N(G_1)=\{u\in V(G)| \ uv \in E(G) \mbox{ for some } v\in V(G_1)\}$ be the {\em neighbourhood} of $G_1$ in $G$. Specifically, if $G_1=\{v\}$, let $N(v)$ denote the neighbourhood of $v$ in $G$. For $X,Y\subseteq V(G)$, a {\em shortest ($X$,$Y$)-path} is a shortest path in $G$ between some vertex of $X$ and some vertex of $Y$. If $X$ contains a vertex $u$ only, then we write $u$ instead of $\{u\}$. Let $d_G(u,v)$ denote the {\em distance between two vertices} $u$ and $v$ in $G$, i.e., the length of a shortest $(u,v)$-path. Analogously, $d_G(X,Y)$ denote the {\em distance between $X$ and $Y$}, i.e., the length of a shortest $(X,Y)$-path in $G$.

Let $P_{\infty}$ denote the two-way infinite path, i.e., $V(P_{\infty})=\mathbb{Z}$ and $E(P_{\infty})=\{i\ i+1|\ i\in\mathbb{Z}\}$ and let $P^{+}_{\infty}$ denote the one-way infinite path, i.e., $V(P^{+}_{\infty})=\mathbb{N}$ and $E(P^{+}_{\infty})=\{i\ i+1|\ i\in\mathbb{N}\}$.

In our proofs we will use the following statement presented by Goddard et al. in  \cite{God}.

\begin{propositionA}{\cite{God}} \label{preliminaries}
Let $k$ be a positive integer. Then
\begin{enumerate}
\item[i)] Every cycle has packing chromatic number at most $4$;
\item[ii)] There is a packing colouring of $P_{\infty}$ with colours $\{k,k+1,\dots, \, 3k+2\}$;
\item[iii)]  If $k\ge 34$, then there is a packing colouring of $P_{\infty}$ with colours $\{k,k+1,\dots, \, 3k-1\}$.
\end{enumerate}
\end{propositionA}

\section{$2$-connected subcubic outerplanar graphs with the weak dual a path} 

In the different proofs of this paper, we say that we denote the vertices of a path of order $n$ by $x^1,x^2,\ldots,x^n$ in an ordering starting by $x$ and finishing by $y$, in the case $x^i$ and $x^{i+1}$ are adjacent, for $1\le i\le n-1$, $x^i$ and $x$ denote the same vertex and $x^n$ and $y$ denote the same vertex.

The following observation will be used in order to construct some useful shortest path in $2$-connected subcubic outerplanar graphs. Moreover, it gives a description of $2$-connected outerplanar graphs with the weak dual of these graphs. 

\begin{observation}
Let $G$ be a $2$-connected outerplanar graph that is not a cycle. Then $G$ contains at least two end faces. Moreover, if $G$ contains no internal face, then $G$ has exactly two end faces.
\end{observation}

\begin{proof}
Considering the weak dual, $\mathcal{T}_G$ is a tree by connectedness of $G$. Since every nontrivial tree has at least two leaves and each leaf of $\mathcal{T}_G$ corresponds to some end face of $G$, each $2$-connected outerplanar graph that is not a cycle contains at least two end faces. In particular, if $G$ has no internal face, then $\mathcal{T}_G$ is a path (the converse also holds), implying that $G$ has exactly two end faces.
\end{proof}

We begin this section with the following lemma. This lemma will be used in Sections 3 and 4.

\begin{lemma} \label{lemma}
Let $G$ be a $2$-connected subcubic outerplanar graph with no internal face. Let $x,y$ be a pair of vertices of degree 2 in $G$ such that $x$ belongs to one of the end faces of $G$ and $y$ to the other one, and let $P$ be a shortest $x,y$-path in $G$. Then there exists a packing colouring of $G$ such that the vertices of $V(G)\setminus V(P)$ are coloured with colours from $\{1,2,3,4\}$.
\end{lemma}

\begin{proof}
Let $C$ denote the boundary cycle of $G$. If there is no chord in $G$, then $G$ is a cycle and, by Proposition \ref{preliminaries}.i), $\chi_{\rho}(G)\leq 4$.

Thus we may assume that $C$ contains some chords in $G$. Note that $C-P$ is not necessarily connected, but, since there is no internal face, each component of $C-P$ is an induced path of $G$. Let $D_i$, $i=1,\dots, \,k$, denote the components of $C-P$ in an ordering from $x$ to $y$ (i.e., for $i<j$, $D_i$ has a neighbour in $P$ that is closest to $x$ than any neighbour of $D_j$ in $P$), and let $l_i$ denote 
the length of $D_i$. We further denote the vertices of each $D_i$ by $x_i^1, x_i^2, \dots, x_i^{l_i}$ in 
an ordering starting from a vertex of $D_i$ which is closest to $x$ in $D_i$. The described structure is shown in Fig. \ref{figL5}, where the thick $x,y$-path depicts $P$. 

\begin{figure}
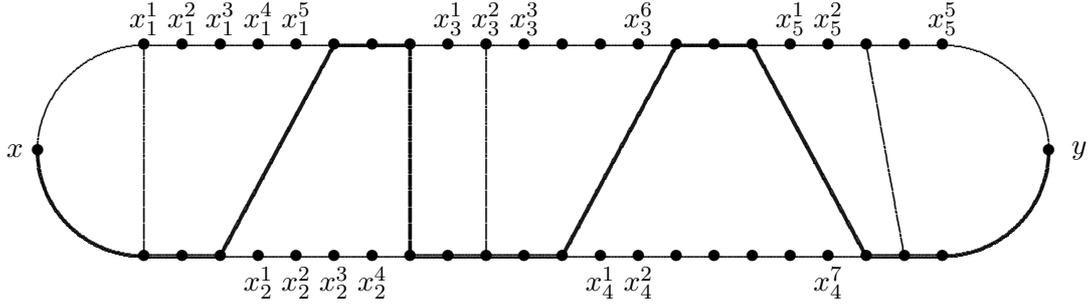
 
\beginpicture
\setcoordinatesystem units <1mm,0.7mm>
\setplotarea x from 0 to 155, y from 0 to 55
\ellipticalarc axes ratio 1:1 180 degrees from 25 45 center at 25 25
\ellipticalarc axes ratio 1:1 -90 degrees from 25 4.8 center at 25 25
\ellipticalarc axes ratio 1:1 -90 degrees from 25 5.2 center at 25 25
\ellipticalarc axes ratio 1:1 180 degrees from 130 5 center at 130 25
\ellipticalarc axes ratio 1:1 90 degrees from 130 5.2 center at 130 25
\ellipticalarc axes ratio 1:1 90 degrees from 130 4.8 center at 130 25
\plot 25 45  130 45 /
\plot 25 5  130 5 /
\plot 25 4.8  35.1 4.8  50.1 44.8  60 44.8  59.85  45  59.85  5  60 4.8  80.1 4.8  95.1 44.8  104.9 44.8  119.9 4.8  130 4.8 /
\plot 25 5.2  34.9 5.2  49.9 45.2  60 45.2  60.15  45  60.15  5  60 5.2  79.9 5.2  94.9 45.2  105.1 45.2  120.1 5.2  130 5.2 /

\plot 25 5  25 45 /
\plot 35 5  50 45 /
\plot 60 5  60 45 /
\plot 70 5  70 45 /
\plot 80 5  95 45 /
\plot 120 5  105 45 /
\plot 125 5  120 45 /
\put{$x$} at 8 25
\put{$y$} at 148 25
\put{$\bullet$} at 11 25
\put{$\bullet$} at 144 25  
\put{$\bullet$} at 25 45  
\put{$\bullet$} at 30 45  
\put{$\bullet$} at 35 45  
\put{$\bullet$} at 40 45  
\put{$\bullet$} at 45 45  
\put{$\bullet$} at 50 45  
\put{$\bullet$} at 55 45  
\put{$\bullet$} at 60 45  
\put{$\bullet$} at 65 45  
\put{$\bullet$} at 70 45  
\put{$\bullet$} at 75 45  
\put{$\bullet$} at 80 45  
\put{$\bullet$} at 85 45  
\put{$\bullet$} at 90 45  
\put{$\bullet$} at 95 45  
\put{$\bullet$} at 100 45  
\put{$\bullet$} at 105 45  
\put{$\bullet$} at 110 45  
\put{$\bullet$} at 115 45  
\put{$\bullet$} at 120 45  
\put{$\bullet$} at 125 45  
\put{$\bullet$} at 130 45  
\put{$\bullet$} at 25 5  
\put{$\bullet$} at 30 5  
\put{$\bullet$} at 35 5  
\put{$\bullet$} at 40 5  
\put{$\bullet$} at 45 5  
\put{$\bullet$} at 50 5  
\put{$\bullet$} at 55 5  
\put{$\bullet$} at 60 5  
\put{$\bullet$} at 65 5  
\put{$\bullet$} at 70 5  
\put{$\bullet$} at 75 5  
\put{$\bullet$} at 80 5  
\put{$\bullet$} at 85 5  
\put{$\bullet$} at 90 5  
\put{$\bullet$} at 95 5  
\put{$\bullet$} at 100 5  
\put{$\bullet$} at 105 5  
\put{$\bullet$} at 110 5  
\put{$\bullet$} at 115 5  
\put{$\bullet$} at 120 5  
\put{$\bullet$} at 125 5  
\put{$\bullet$} at 130 5  
\put{$x_1^1$} at 25 50
\put{$x_1^2$} at 30 50
\put{$x_1^3$} at 35 50
\put{$x_1^4$} at 40 50
\put{$x_1^5$} at 45 50
\put{$x_2^1$} at 40 0
\put{$x_2^2$} at 45 0
\put{$x_2^3$} at 50 0
\put{$x_2^4$} at 55 0
\put{$x_3^1$} at 65 50
\put{$x_3^2$} at 70 50 
\put{$x_3^3$} at 75 50
\put{$x_3^6$} at 90 50
\put{$x_4^1$} at 85 0
\put{$x_4^2$} at 90 0
\put{$x_4^7$} at 115 0
\put{$x_5^1$} at 110 50
\put{$x_5^2$} at 115 50
\put{$x_5^5$} at 130 50

\endpicture
\caption{Structure of $G$ in Lemma \ref{lemma}.}
\label{figL5}
\end{figure}

We colour each component $D_i$ of $C-P$ with a pattern $1,2,1,3$ starting from $x_i^1$ ($i=1,2,\dots, k$), i.e., for each odd $j$, $x_i^j$ is coloured with colour $1$, for each $j$ divisible by $4$, $x_i^j$ obtains colour $3$, and, for each even $j$ not divisible by $4$, we colour vertex $x_i^j$ with colour $2$, $j=1,2,\dots, l_i$.
 We denote by $\chi$ the defined colouring. Since $P$ is a shortest path and $G$ is subcubic, we are going to show that there is no collision between any pair of vertices coloured with colour $1$ or $2$, respectively. Suppose to the contrary that there is a pair of clashing vertices $a$ and $b$ coloured with colour $1$. Clearly $a$ and $b$ belong to the same component $D_i$ of $C-P$ and $ab\in E(G)\setminus E(C)$ by the definition of $\chi$. But then we get a contradiction with the fact that $D_i$ is an induced path of $G$. Now suppose that there is a pair of clashing vertices $a$ and $b$ coloured with colour $2$, i.e., $d_G(a,b)\leq 2$. Again, $a$ and $b$ must belong to the same component $D_i$ of $C-P$, otherwise $d_G(a,b)>2$ and we obtain a contradiction. Analogously as for colour $1$ we can show that $ab\not\in E(G)$. Thus $a$ and $b$ must have a common neighbour $c$ in $G$. From the definition of $\chi$ and also because $D_i$ is an induced path in $G$, $c\not\in V(D_i)$, and $c\not\in P$ since $G$ is subcubic, a contradiction with the existence of $c$.

Therefore the only possible collision in the defined colouring $\chi$ could be between vertices coloured with colour $3$. Analogously as for colours $1$ and $2$, any pair of clashing vertices $a$ and $b$ cannot be at distance one or two apart. Therefore any such collision happens for $a$ and $b$ with $d_G(a,b)=3$. We will check and modify collisions in the defined colouring $\chi$ of the components $D_1,D_2,\dots, D_k$ of $C-P$ one-by-one starting from $D_1$ and from the vertex $x_i^1$ in each $D_i$. Note that, in each step of the modification process, we check the modified colouring, not the original one. The following possible collisions can occur:

\begin{mathitem} 
\item[{\sl Case 1:}] $a$ and $b$ belong to different components $D_i$ and $D_j$ of $C-P$, $i,j\in\{ 1,2,\dots, k \}$. Since $d_G(a,b)=3$, (up to a symmetry) $a=x_i^{l_i}$, 
$b$ belongs to the chord of $D_j\cup P$ which is closest to $a$, and $j=i+1$, otherwise we get $d_G(a,b)>3$. Then we can modify the colouring $\chi$ of $D_j$ by recolouring vertices $x_j^s$, $s=3,\dots, l_j-1$, with $\chi(x_j^s):= \chi(x_j^{s+1})$ and we set $\chi(x_j^{l_j})\in \{1,2,3\}$ depending on the continuation of the pattern $1,2,1,3$ in $D_j$.

\item[{\sl Case 2:}] $a$ and $b$ belong to the same component $D_i$. Since $d_G(a,b)=3$, $a$ and $b$ must belong to consecutive chords of $P\cup D_i$ and there is no vertex between these chords on $P$. We call such a pair of vertices coloured with colour $3$ a {\em critical pair}. Consider a critical pair $a$ and $b$ such that $a=x_i^m$, $b=x_i^n$, $m<n$, and that there is no critical pair $a'$ and $a$ with $a'=x_i^o$, $o<m$. Then we modify the colouring of the vertices of $D_i$ starting at vertex $a$ by $\underline{3},1,4,1,\underline{2},1,3,1,2,\dots$ instead of $\underline{3},1,2,1,\underline{3},1,2,1,3 \dots$, i.e., we recolour the vertex $x_i^{m+2}$ with colour $4$ and we switch colours $2$ and $3$ of the vertices $x_i^j$ for even $j>m+2$. Note that the underlined colours represent the critical pair $a$ and $b$. It is easy to verify that vertices coloured with colour $4$ are mutually at distance more than $4$ apart, implying that there is no collision between any pair of vertices coloured with colour $4$.
\end{mathitem}
After these modifications we obtain a colouring of all the vertices of $G-P$ with colours $\{1,2,3,4\}$ satisfying the distance constraints of a packing colouring.

\end{proof}

\begin{theorem} \label{out1}
If $G$ is a $2$-connected subcubic outerplanar graph with no internal face, then $\chi_{\rho}(G)\leq 15$.
\end{theorem}

\begin{proof}
Let $x,y$ be any pair of vertices of degree 2 in $G$ belonging to distinct end faces of $G$. Let $P$ be a shortest $x,y$-path in $G$. By Lemma~\ref{lemma}, the vertices of $V(G)\setminus V(P)$ can be coloured with colours from $\{1,2,3,4\}$. 
Then the colouring can be completed in a packing 15-colouring of $G$ by colouring the vertices along the path $P$ starting at $x$ and using a packing colouring of the infinite path (since $P$ is a shortest path in $G$, then the distance between any pair of vertices of $P$ is the same on $P$ and on $G$). For this, we repeat the following pattern with colours from $\{5,\ldots,15\}$ of length 36 along the vertices of $P$ starting at $x$:
\begin{equation}\label{pat5-15}
5,6,7,9,13,12,5,8,6,10,7,11,5,9,14,6,8,15,5 ,7,13,10,6,11,5,8,9,7,12,6,5,14,10,15,8,11
\end{equation}

It is easy to check that any two colours $i$ in this repeating sequence are separated by at least $i$ integers.
\end{proof}

Note that the previous pattern was found by a computer search. 

\section{Asymptotic results for subcubic outerplanar graphs}

The main goal of this paper is to study
the finiteness of the packing chromatic number of subcubic outerplanar graphs, i.e, we ask whether the packing chromatic number of an outerplanar graph with maximum degree at most $3$ depends on the order of the graph or not. In this section we prove that, for any $2$-connected outerplanar graph with a fixed number of internal faces and for any connected outerplanar
graph with a fixed number of faces, the packing chromatic number does not depend on the order of the graph.

We begin this section by proving the following useful lemma that will also be used in Section 4. We recall that the weak dual of a $2$-connected outerplanar graph is a tree and that $u_F$ is the vertex of the weak dual corresponding to the face $F$.

\begin{lemma}\label{out2}
 There exists a packing colouring of $P^{+}_{\infty}$ with colours $\{5,\ldots,15\}$ such that the first vertex along the path is at distance at least $\lceil (i-5)/2 \rceil$ of any vertex of colour $i$.

\end{lemma}

\begin{proof}
By considering Pattern~\eqref{pat5-15} from the proof of Theorem~\ref{out1} starting at the first vertex of the path (the vertex of degree $1$), we can easily check that the first six vertices of $P^{+}_{\infty}$ satisfy the property. Since the colours used in Pattern~\eqref{pat5-15} are bounded by $15$, the other vertices (other than the first six vertices) satisfy the property as well.
\end{proof}

For positive integers $i$, $j$ and $k$, let $r^k_{i,j}\in \mathbb{Z}$ such that $ r^k_{i,j}\equiv i-j \pmod{k}$ with minimum absolute value. The value $|r^k_{i,j}|$ corresponds to the distance between two vertices $i$ and $j$ in a cycle $C_k$ with vertex set $\{0,\ldots, k-1\}$ (the vertices are enumerated along the cycle). 
A subset of vertices $A$ of a graph $G$ is a \emph{cycle-distance-preserved} set if there exists an ordering $v_{A}^{0},\ldots,v_{A}^{|A|-1}$ of the vertices of $A$ satisfying $d_G(v^j_{A},v^{j'}_{A})\ge |r^{|A|}_{j,j'}|$, for integers $0\le j<j'\le |A|-1$.

\begin{lemma}\label{acycle}
For any positive integers $k$ and $n>2$, there exists a packing colouring of the cycle $C_n$ with colours from $\{k,\ldots,6k+4  \}$.
%Let $G$ be graph and let $A\subseteq V(G)$ be a subset preserving cycle distance. The vertices of $A$ can be packing-coloured with colours $\{\ell,\ldots,6\ell+4  \}$, for any integer $\ell$.
\end{lemma}
\begin{proof} Let $C_n$ be a cycle of length $n$.
First, if $n\le 5k+5$, then we can colour each vertex of $C_n$ with a different colour from $\{k,\ldots,6k+4  \}$.

Second, if $5k+5 <n\le 6k+5$, then we colour $3k$ consecutive vertices of $C$ with colours $k,\ldots, 3k-1,k,\dots, 2k-1$, and colour the remaining $n-3k\le 3k+5$ vertices of $C_n$ with mutually distinct colours from $\{ 3k, \ldots, 6k+4\}$.

Third, suppose $n> 6k+5$. By Proposition \ref{preliminaries}.ii), we can colour $P_{\infty}$ with colours from $\{k,\ldots, 3k+2\}$.
Let $P'$ be any subpath of $C_n$ on $3k+2$ vertices. Since the distance between the two ends of $C_n-P'$ is at least $3k+3$ in $C_n-P'$ and exactly $3k+3$ in $C_n$, we can colour the vertices of $C_n-P'$ with the colours $\{k,\ldots, 3k+2\}$ (using  Proposition \ref{preliminaries}.ii) ) and the vertices of $P'$ with mutually distinct colours from $\{ 3k +3,\ldots, 6k+ 4 \} $.
 %Hence, we can use the colouring of this cycle $C$ in order to colour the vertices of $A$ with colours from $\{\ell,\ldots,6\ell+4  \}$.
\end{proof}

A subset of vertices $A$ of a graph $G$ is \emph{decomposable} into $r$ cycle-distance-preserved  sets if there exist $r$ sets of vertices $A_1$, $A_2$, $\ldots$, $A_r$, such that $A_1\cup \ldots \cup A_r=A$ and for each integer $i$, $A_i$ is a cycle-distance-preserved set. The following lemma will be useful in order to prove Theorems \ref{k-inner face structure} and \ref{k-face structure}.

\begin{lemma}\label{disjoint cycles}
Let $G$ be a graph and let $A\subseteq V(G)$ be a subset decomposable into $r$ cycle-distance-preserved sets. The vertices of $A$ can be packing-coloured with colours $\{k,\ldots,6^{r} (k+1)-2  \}$, for any positive integer $k$.
\end{lemma}

\begin{proof}
We proceed by induction on $r$. For $r=1$, since $A$ is a cycle-distance-preserved set, by Lemma \ref{acycle}, we can colour the vertices of $A$ with colours $\{k,\ldots,6k+4  \}$.
Now suppose that a subset $A\subset V(G)$ is decomposable into $r+1$ cycle-distance-preserved sets. Using induction hypothesis we can colour the vertices of $A_1$, $\ldots$, $A_r$ with colours $\{k,\ldots,6^{r} (k+1)-2 \}$.
For the vertices of $A_{r+1}$, by Lemma \ref{acycle}, we can use colours $\{6^{r} (k+1)-1, \ldots, k '\}$, where $k'= 6( 6^{r} (k+1) -1)+4=  6^{r+1} (k+1)-2$. Note that we do not need to change colours of the vertices from $\cup_{i=1}^r (A_i \cap A_{r+1})$ (in the case it is not empty).
\end{proof}

%The following theorem is one of our main results.
%It can be used in order to prove that some subcubic outerplanar graphs have finite packing chromatic number. 
The following theorem is one of our main results. It shows that the packing chromatic number of a 2-connected subcubic outerplanar graph with bounded number of internal faces is bounded by a constant.

\begin{theorem} \label{k-inner face structure}
If $G$ is a $2$-connected subcubic outerplanar graph with $r$ internal faces, then $\chi_{\rho}(G)\le  17\times  6^{3r}-2$.
\end{theorem}

\begin{proof}
Let $F_1, \ldots, F_r$ denote the $r$ distinct internal faces of $G$, and $u_{F_1}, \dots u_{F_r}$ the corresponding vertices of $\mathcal{T}_G$ (note that each $u_{F_i}$ has degree at least $3$ in $\mathcal{T}_G$). Figure~\ref{figadd1} illustrates a subcubic outerplanar graph and its dual.
Then, removing the vertices $u_{F_1}, \dots, u_{F_r}$ from $\mathcal{T}_G$, we obtain a union of disjoint paths. 
The connected components with one end vertex adjacent in $\mathcal{T}_G$ to $u_{F_i}$ and the other end vertex adjacent in $\mathcal{T}_G$ to $u_{F_j}$ are denoted by $U_{i,j}$. 
In the case $u_{F_i}$ and $u_{F_j}$ are adjacent, $U_{i,j}$ is not defined (all the vertices in the face "between" $F_i$ and $F_j$ are colored in Step 3).
For any $u_{F_i}$, $i=1,\dots, r$, the paths with one end vertex of degree $1$ and the other one adjacent to $u_{F_i}$, are denoted by $U_{i}^{1}, \ldots, U_{i}^{\ell_i}$, where $\ell_i$ is the number of such paths for $u_{F_i}$. Note 
that some of the paths $U_{i,j}$,$U_{i}^{q}$ may be trivial or empty. 

\begin{figure}[ht]
\begin{center}
\begin{tikzpicture}[scale=0.8]
\draw (-4,0) -- (4,0);
\draw (-4,1.5) -- (4,1.5);
\draw (-1,0) -- (0,1.5);
\draw (1,0) -- (1,1.5);
\draw (-2,0) -- (-3,1.5);
\draw (2,1.5) -- (3,0);
\draw (4,0) -- (4,1.5);
\draw (-4,0) -- (-4,1.5);

\draw (4,0) -- (4.8,-0.7);
\draw (6.2,-0.7) -- (4.8,-0.7);
\draw (6.2,-0.7) -- (7,0);
\draw (7,0) -- (7,1.5);
\draw (7,1.5) -- (6.2,2.2);
\draw (6.2,2.2) -- (4.8,2.2);
\draw (4.8,2.2) -- (4,1.5);

\draw (-4,0) -- (-4.8,-0.7);
\draw (-6.2,-0.7) -- (-4.8,-0.7);
\draw (-6.2,-0.7) -- (-7,0);
\draw (-7,0) -- (-7,1.5);
\draw (-7,1.5) -- (-6.2,2.2);
\draw (-6.2,2.2) -- (-4.8,2.2);
\draw (-4.8,2.2) -- (-4,1.5);

\draw (-4.8,2.2) -- (-4.8,3);
\draw (-6.2,2.2) -- (-6.2,3);
\draw (-4.8,3) -- (-5.5,3.7);
\draw (-6.2,3) -- (-5.5,3.7);
\draw (-4.8,3) -- (-6.2,3);

\draw (4.8,2.2) -- (5.5,3);
\draw (6.2,2.2) -- (5.5,3);
\draw (7,1.5) -- (11,1.5);
\draw (7,0) -- (11,0);
\draw (8,0) -- (9,1.5);
\draw (11,0) -- (10,1.5);
\draw (11,0) -- (11,1.5);

\draw (-7,0) -- (-7.3,-1);
\draw (-6.2,-0.7) -- (-7.3,-1);

\draw (4.8,-0.7) -- (4.8,-1.6);
\draw (6.2,-0.7) -- (6.2,-1.6);
\draw (4.8,-1.6) -- (6.2,-1.6);

\draw[ultra thick, style=dashed,color=red] (0,1.5) -- (2,1.5);
\draw[ultra thick, style=dashed,color=red] (-1,0) -- (0,1.5);
\draw[ultra thick, style=dashed,color=red] (-1,0) -- (-2,0);
\draw[ultra thick, style=dashed,color=red] (-2,0) -- (-3,1.5);
\draw[ultra thick, style=dashed,color=red] (2,1.5) -- (3,0);

\draw[ultra thick, style=dashed,color=red] (-5.5,3.7) -- (-6.2,3);

\draw[ultra thick, style=dashed,color=red] (8,0) -- (9,1.5);
\draw[ultra thick, style=dashed,color=red] (9,1.5) -- (11,1.5);

\node at (-4,0) [circle,draw=black,fill=black,scale=0.5]{};
\node at (-3,0) [circle,draw=black,fill=black,scale=0.5]{};
\node at (-2,0) [circle,draw=black,fill=black,scale=0.5]{};
\node at (-1,0) [circle,draw=black,fill=black,scale=0.5]{};
\node at (0,0) [circle,draw=black,fill=black,scale=0.5]{};
\node at (1,0) [circle,draw=black,fill=black,scale=0.5]{};
\node at (2,0) [circle,draw=black,fill=black,scale=0.5]{};
\node at (4,0) [circle,draw=black,fill=black,scale=0.5]{};
\node at (3,0) [circle,draw=black,fill=black,scale=0.5]{};
\node at (-4,1.5) [circle,draw=black,fill=black,scale=0.5]{};
\node at (-3,1.5) [circle,draw=black,fill=black,scale=0.5]{};
\node at (-2,1.5) [circle,draw=black,fill=black,scale=0.5]{};
\node at (-1,1.5) [circle,draw=black,fill=black,scale=0.5]{};
\node at (0,1.5) [circle,draw=black,fill=black,scale=0.5]{};
\node at (1,1.5) [circle,draw=black,fill=black,scale=0.5]{};
\node at (2,1.5) [circle,draw=black,fill=black,scale=0.5]{};
\node at (4,1.5) [circle,draw=black,fill=black,scale=0.5]{};
\node at (3,1.5) [circle,draw=black,fill=black,scale=0.5]{};

\node at (4.8,-0.7) [circle,draw=black,fill=black,scale=0.5]{};
\node at (6.2,-0.7) [circle,draw=black,fill=black,scale=0.5]{};
\node at (7,0) [circle,draw=black,fill=black,scale=0.5]{};
\node at (7,1.5) [circle,draw=black,fill=black,scale=0.5]{};
\node at (6.2,2.2) [circle,draw=black,fill=black,scale=0.5]{};
\node at (4.8,2.2) [circle,draw=black,fill=black,scale=0.5]{};

\node at (-4.8,-0.7) [circle,draw=black,fill=black,scale=0.5]{};
\node at (-6.2,-0.7) [circle,draw=black,fill=black,scale=0.5]{};
\node at (-7,0) [circle,draw=black,fill=black,scale=0.5]{};
\node at (-7,1.5) [circle,draw=black,fill=black,scale=0.5]{};
\node at (-6.2,2.2) [circle,draw=black,fill=black,scale=0.5]{};
\node at (-4.8,2.2) [circle,draw=black,fill=black,scale=0.5]{};

\node at (5.5,3) [circle,draw=black,fill=black,scale=0.5]{};

\node at (-4.8,3) [circle,draw=black,fill=black,scale=0.5]{};
\node at (-6.2,3) [circle,draw=black,fill=black,scale=0.5]{};
\node at (-5.5,3.7) [circle,draw=black,fill=black,scale=0.5]{};

\node at (8,1.5) [circle,draw=black,fill=black,scale=0.5]{};
\node at (9,1.5) [circle,draw=black,fill=black,scale=0.5]{};
\node at (10,1.5) [circle,draw=black,fill=black,scale=0.5]{};
\node at (11,1.5) [circle,draw=black,fill=black,scale=0.5]{};
\node at (8,0) [circle,draw=black,fill=black,scale=0.5]{};
\node at (9,0) [circle,draw=black,fill=black,scale=0.5]{};
\node at (10,0) [circle,draw=black,fill=black,scale=0.5]{};
\node at (11,0) [circle,draw=black,fill=black,scale=0.5]{};

\node at (-7.3,-1) [circle,draw=black,fill=black,scale=0.5]{};

\node at (4.8,-1.6) [circle,draw=black,fill=black,scale=0.5]{};
\node at (6.2,-1.6) [circle,draw=black,fill=black,scale=0.5]{};

\node at (-5.5,0.75) {\small{$F_1$}};
\node at (5.5,0.75) {\small{$F_2$}};

\node at (-5.5,3.25) {\tiny{$\hat F_{1}^{1}$}};
\node at (-6.8,-0.55) {\tiny{$\hat F_{1}^{2}$}};

\node at (5.5,2.55) {\tiny{$\hat F_{2}^{1}$}};
\node at (10.7,1) {\tiny{$\hat F_{2}^{2}$}};
\node at (5.5,-1.2) {\tiny{$\hat F_{2}^{3}$}};

\node at (-5.2,3.8) {\tiny{$y^1_1$}};
\node at (-6.9,-1.1) {\tiny{$y^2_1$}};

\node at (5.8,3.1) {\tiny{$y^1_2$}};
\node at (11.2,1.8) {\tiny{$y^2_2$}};
\node at (6.5,-1.8) {\tiny{$y^3_2$}};

\node at (8.2,-0.3) {\tiny{$(p_2^2)_1$}};
\node at (9.2,1.8) {\tiny{$(p_2^2)_2$}};
\node at (10.2,1.8) {\tiny{$(p_2^2)_3$}};

\node at (-6.7,3) {\tiny{$(p_1^1)_1$}};

\node at (-2.8,1.8) {\tiny{$p_{1,2}^1$}};
\node at (-1.8,-0.3) {\tiny{$p_{1,2}^2$}};
\node at (-0.8,-0.3) {\tiny{$p_{1,2}^3$}};
\node at (0.2,1.8) {\tiny{$p_{1,2}^4$}};
\node at (1.2,1.8) {\tiny{$p_{1,2}^5$}};
\node at (2.2,1.8) {\tiny{$p_{1,2}^6$}};
\node at (3.2,-0.3) {\tiny{$p_{1,2}^7$}};

\draw[dotted,color=blue] (5.5,0.75) circle (2.75cm);
\draw[color=red] (2,1.5) circle (0.7cm);
\draw[color=red] (9,1.5) circle (0.7cm);
\draw[dotted,color=blue] (-5.5,0.75) circle (2.75cm);
\draw[color=red] (-2,0) circle (0.7cm);
\draw[color=red] (-5.5,3.7) circle (0.7cm);
\end{tikzpicture}
\end{center}
\begin{center}
\begin{tikzpicture}[scale=1.6]
\draw (-3,0) -- (3,0);

\draw[->,>=stealth] (-3,0) -- (-3.25,0.25);
\draw (-3.25,0.25) -- (-3.5,0.5);

\draw[->,>=stealth] (-3.5,0.5) -- (-3.75,0.75);
\draw(-3.75,0.75) -- (-4,1);
\draw[->,>=stealth] (-3,0) -- (-3.25,-0.25);
\draw (-3.25,-0.25) -- (-3.5,-0.5);

\draw (-3,0) -- (5,0);

\draw[->,>=stealth] (-3,0) -- (-2.5,0);
\draw[->,>=stealth] (-2,0) -- (-1.5,0);

\draw[->,>=stealth] (-1,0) -- (-0.5,0);
\draw[->,>=stealth] (0,0) -- (0.5,0);
\draw[->,>=stealth] (1,0) -- (1.5,0);
\draw[->,>=stealth] (2,0) -- (2.5,0);
\draw[->,>=stealth] (3,0) -- (3.35,0);
\draw[->,>=stealth] (3.7,0) -- (4.05,0);
\draw[->,>=stealth] (4.4,0) -- (4.75,0);
\draw[->,>=stealth] (3,0) -- (3.25,0.25);
\draw (3.25,0.25) -- (3.5,0.5);
\draw[->,>=stealth] (3,0) -- (3.25,-0.25);
\draw(3.25,-0.25) -- (3.5,-0.5);

\node at (-3,0) [circle,draw=black,fill=black,scale=0.5]{};
\node at (-2,0) [circle,draw=black,fill=black,scale=0.5]{};
\node at (-1,0) [circle,draw=black,fill=black,scale=0.5]{};
\node at (3,0) [circle,draw=black,fill=black,scale=0.5]{};
\node at (2,0) [circle,draw=black,fill=black,scale=0.5]{};
\node at (1,0) [circle,draw=black,fill=black,scale=0.5]{};
\node at (0,0) [circle,draw=black,fill=black,scale=0.5]{};
\node at (-3.5,-0.5) [circle,draw=black,fill=black,scale=0.5]{};
\node at (-3.5,0.5) [circle,draw=black,fill=black,scale=0.5]{};
\node at (-4,1) [circle,draw=black,fill=black,scale=0.5]{};

\node at (3.5,-0.5) [circle,draw=black,fill=black,scale=0.5]{};
\node at (3.5,0.5) [circle,draw=black,fill=black,scale=0.5]{};
\node at (3.7,0) [circle,draw=black,fill=black,scale=0.5]{};
\node at (4.4,0) [circle,draw=black,fill=black,scale=0.5]{};
\node at (5.1,0) [circle,draw=black,fill=black,scale=0.5]{};

\node at (-2.7,-0.3) {\small{$u_{F_1}=z$}};
\node at (2.9,-0.3) {\small{$u_{F_2}$}};
\node at (0.65,-0.5) {\small{$\overrightarrow{\mathcal{T}_{G}}$}};

\end{tikzpicture}
\end{center}
\caption{\label{figadd1} A subcubic outerplanar graph $G$ with two internal faces (on the top) and its oriented dual (on the bottom) (dashed lines : paths $P_{1,2}$, $P^{1}_1$ and $P^2_2$; vertex in dotted circle : vertex in $B_1\cup B_2$ or $V(F_1)\cup V(F_2)$; vertex in simple circle: vertex in $D_1\cup D_2$).}
\end{figure}
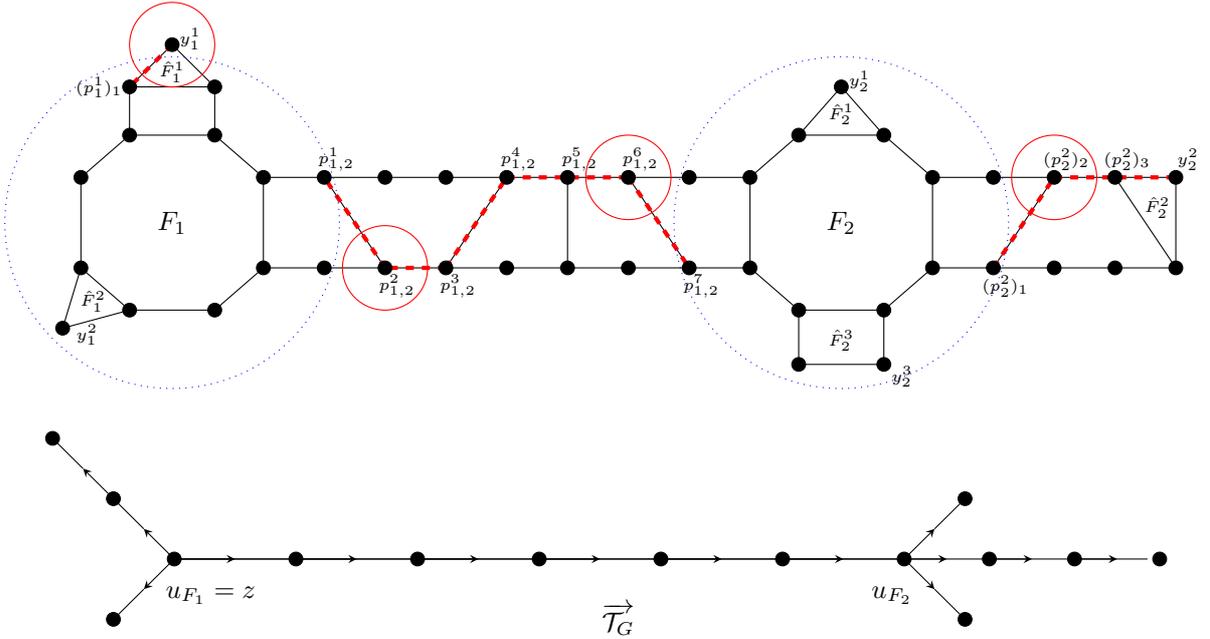

%In this proof, we will consider differently the non internal faces depending on the kind of path to which they correspond in $\mathcal{T}_G$.

%Let $B$ be the set of vertices of $G$ at distance at most one from an internal face. 
For any $i$, $1\le i\le r$, let $B_i=N(F_i)\setminus \bigcup_{i=1}^r V(F_i)$, and let $B=\bigcup_{i=1}^r (F_i\cup B_i)$.
Let $i$ and $q$ be integers such that $1\le i\le r$, $1\le q\le \ell_i$. Consider an end face $\hat F_{i}^{q}$ in $G$ corresponding to an end vertex of $V(U_{i}^{q})$ of degree 1 in $\mathcal{T}_G$. We denote by $y_i^{q}$ a vertex of $\hat F_{i}^{q}$ of degree $2$ (note that such a vertex always exists) and by $P_i^{q}$ a shortest $(B_i,y_i^{q})$-path in $G$. Let $(p_{i}^{q})_1,(p_{i}^{q})_2,\ldots$ denote vertices of $P_i^{q}$ in an ordering starting from the vertex of $B_i$. 

Choose any vertex $z$ of $\mathcal{T}_G$ such that $z$ has degree $3$ in $\mathcal{T}_G$ and let $\overrightarrow{\mathcal{T}_G}$ be the digraph obtained from $\mathcal{T}_G$ by replacing each edge $uv\in E(\mathcal{T}_G)$ satisfying $d_{\mathcal{T}_G}(u,z)<d_{\mathcal{T}_G}(v,z)$ with an arc from $u$ to $v$.

Now we consider the paths $U_{i,j}$ in $\mathcal{T}_G$. Let $i,j$ be positive integers such that $1\le i<j\le r$ and $U_{i,j}$ is defined and has length at least one. Let $P_{i,j}$ be a shortest $(B_i, B_j)$-path and $p_{i,j}^1,p_{i,j}^2 \ldots,p_{i,j}^{l_{i,j}}$ its vertices in an ordering starting from the vertex of $B_i$, if $d_{\mathcal{T}_G}(u_{F_i},z)<d_{\mathcal{T}_G}(u_{F_j},z)$, or from the vertex of $B_j$ otherwise (by $l_{i,j}$ we mean the order of $P_{i,j}$).
The upper part of Figure \ref{figadd1} illustrates the notations used in this proof.

Let $P =\left(\bigcup\limits_{1\le i\le r} \bigcup\limits_{1\le q\le \ell_i } V(P_{i}^{q})\right)\,\, \bigcup \,\, \left(\bigcup\limits_{1\le i<j\le k} V(P_{i,j})\right)$.
%The proof is organized as follows. First, we will colour the vertices of $V(G)\setminus (B\cup P)$ with colours $\{1,2,3,4\}$. Second, we will colour some vertices of $P$. Third, we will colour the remaining uncoloured vertices of $G$.

\begin{mathitem}

\item[{\sl Step 1:}] Colouring the vertices of $V(G)\setminus (B\cup P)$ with colours $\{1,2,3,4\}$.

We colour the vertices of $V(G)\setminus (B\cup P)$ by colouring  each connected component (one by one) of $G- (B\cup P)$ in the same way as in the proof of Lemma \ref{lemma}, i.e., we use the pattern $1,2,1,3$.
Note that the distance between any two vertices from $V(G)\setminus (B\cup P)$ in two different connected components of $G-B$ is at least $5$.
Moreover, we proceed as in the proof of Lemma \ref{lemma} to avoid clashing vertices of colour $3$, i.e., we use colour $4$.

\item[{\sl Step 2:}] Colouring vertices of $P$.

Let $i$, $j$, $i'$ and $q$ be integers such that $U_{i,j}$ and $U_{i'}^{q}$ are defined.

For the vertices of $P_{i'}^{q}$, we use Pattern~\eqref{pat5-15} and Lemma \ref{out2} starting at the vertex $(p_{i'}^{q})_3$. For the vertices of $P_{i,j}$, we use Pattern~\eqref{pat5-15} and Lemma \ref{out2}, starting at the vertex $p_{i,j}^3$ and finishing at the vertex $p_{i,j}^{k_{i,j}-3}$.
Note that every vertex of $V(\overrightarrow{\mathcal{T}_G})$ has in-degree at most one. This property, along with Lemma \ref{out2}, ensure us that a vertex coloured with colour $a$ in $P_{i,j}$, $a\in\{5,\ldots,15\}$, is at distance at least $a+1$ from any other vertex coloured by $a$ in $P_{\bar i,\bar j}$, for $1\le \bar i< \bar j\le r$.

\item[{\sl Step 3:}] Colouring the remaining vertices of $G$.

Let $w_{i,j}$ be a vertex among $\{p_{i,j}^{2} ,p_{i,j}^{l_{i,j}-2}\}$ at distance $2$ from a vertex of $V(F_i)$ (when $U_{i,j}$ is defined).
Let $D_i$ be the  set $\{(p_{i}^{A})_2|\ 1\le A\le \ell_i\}\cup \{w_{i,j}|\ U_{i,j} \text{ is defined,}\ 1\le j\le k\}$. Since the sets $V(F_i)$, $B_i$ and $D_i$ , $1\le i\le r$, are cycle-distance-preserved sets, the set $\bigcup_{i=1}^r V(F_i)\bigcup_{i=1}^r B_i \bigcup_{i=1}^r D_i$ is decomposable into $3r$ cycle-distance-preserved sets. Hence, using Lemma \ref{disjoint cycles}, the remaining uncoloured vertices can be coloured with colours $\{16,\ldots,17\times  6^{3r}-2 \}$.
\end{mathitem}
\end{proof}

Sloper in \cite{Slo} defined an {\em expandable broadcast-colouring} of a complete binary tree $T$ as a colouring $c$ of $V(T)$ with colours $1,2,\dots, 7$ such that:

\begin{mathitem}
\item $\forall u,v\in V(T)$ $c(u)=c(v) \Rightarrow d_T(u,v)>c(u)$,
\item the root $x$ of $T$ has colour $1$,
\item all vertices at even distance from $x$ have colour $1$,
\item every vertex of colour $1$ has at least one child of colour $2$ or $3$,
\item $c(u)=6, c(v)=7 \Rightarrow d_T(u,v)\ge 5$,
\item $c(u)\in \{4,5,6,7\} \Rightarrow$ $u$'s children each have children coloured with $2$ and $3$.
\end{mathitem}

Notice that an expandable broadcast-colouring of a tree is a packing 7-colouring. Sloper has shown that given an expandable colouring of a (complete) binary tree of height $n$, it is possible to create an expandable colouring of a (complete) binary tree of height
$(n + 1)$ by using the colouring for the tree of height $n$ as a basis~\cite{Slo}.
Note that the colouring of a complete binary tree of height $3$ that consists in giving the colours $2$ and $3$ to the two neighbours of $x$, giving the colour $1$ to the vertex at distance $2$ from $x$, and giving colours from $\{2,3,4,5\}$ to the remaining vertices is an expendable broadcast-colouring. This colouring is described in Figure~\ref{figadd2}.

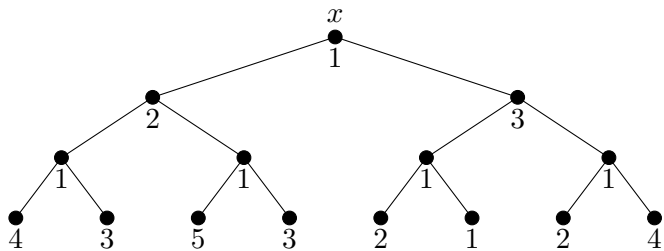
\begin{figure}[ht]
\begin{center}
\begin{tikzpicture}[scale=0.8]
\draw (0,3) -- (-3,2);
\draw (0,3) -- (3,2);
\draw (-3,2) -- (-4.5,1);
\draw (-3,2) -- (-1.5,1);
\draw (3,2) -- (4.5,1);
\draw (3,2) -- (1.5,1);
\draw (-4.5,1) -- (-5.25,0);
\draw (-4.5,1) -- (-3.75,0);
\draw (-1.5,1) -- (-2.25,0);
\draw (-1.5,1) -- (-0.75,0);
\draw (4.5,1) -- (5.25,0);
\draw (4.5,1) -- (3.75,0);
\draw (1.5,1) -- (2.25,0);
\draw (1.5,1) -- (0.75,0);

\node at (0,3) [circle,draw=black,fill=black,scale=0.5]{};
\node at (-3,2) [circle,draw=black,fill=black,scale=0.5]{};

\node at (3,2) [circle,draw=black,fill=black,scale=0.5]{};
\node at (-4.5,1) [circle,draw=black,fill=black,scale=0.5]{};
\node at (-1.5,1) [circle,draw=black,fill=black,scale=0.5]{};
\node at (1.5,1) [circle,draw=black,fill=black,scale=0.5]{};
\node at (4.5,1) [circle,draw=black,fill=black,scale=0.5]{};
\node at (-5.25,0) [circle,draw=black,fill=black,scale=0.5]{};
\node at (-3.75,0) [circle,draw=black,fill=black,scale=0.5]{};
\node at (-2.25,0) [circle,draw=black,fill=black,scale=0.5]{};
\node at (-0.75,0) [circle,draw=black,fill=black,scale=0.5]{};
\node at (5.25,0) [circle,draw=black,fill=black,scale=0.5]{};
\node at (3.75,0) [circle,draw=black,fill=black,scale=0.5]{};
\node at (2.25,0) [circle,draw=black,fill=black,scale=0.5]{};
\node at (0.75,0) [circle,draw=black,fill=black,scale=0.5]{};
\node at (0,3.35) {$x$};
\node at (0,2.65) {$1$};
\node at (-3,1.65) {$2$};
\node at (3,1.65) {$3$};
\node at (-4.5,0.65) {$1$};
\node at (-1.5,0.65) {$1$};
\node at (1.5,0.65) {$1$};
\node at (4.5,0.65) {$1$};
\node at (-5.25,-0.35) {$4$};
\node at (-3.75,-0.35) {$3$};
\node at (-2.25,-0.35) {$5$};
\node at (-0.75,-0.35) {$3$};
\node at (5.25,-0.35) {$4$};
\node at (3.75,-0.35) {$2$};
\node at (2.25,-0.35) {$1$};
\node at (0.75,-0.35) {$2$};

\end{tikzpicture}
\end{center}
\caption{\label{figadd2} An expendable broadcast-colouring of a complete binary tree of height $3$.}
\end{figure}

The following statement is true for a more general class of graphs than in Theorem \ref{k-inner face structure} since it gives an upper bound for all connected outerplanar graphs (not necessarily $2$-connected). However, since the parameter is the number of faces, the bound is weaker than the bound in Theorem \ref{k-inner face structure}.

\begin{theorem} \label{k-face structure}
If $G$ is a connected subcubic outerplanar graph with $r$ (non external) faces, then $\chi_{\rho}(G)\le 9\times 6^{r}-2$.
\end{theorem}
\begin{proof}
Let $F_1, \ldots, F_r$ denote the $r$ (non external) faces of $G$. 
The graph $O=G-\bigcup_{i=1}^r V(F_i)$ consists of components $O_1,O_2, \dots, O_s$ such that each $O_j$ ($j=1,\dots, s$) is a tree. And, since $G$ is subcubic, each $O_j$ is subcubic as well. In the weak dual $\mathcal{T}_G$ of $G$, choose arbitrary vertex $z$ and let $F_z$ denote a face corresponding to $z$ in $G$. We colour the vertices of $G$ in two steps.

%Let $O$ be a set of vertices which belong to the outer face only.
%We denote by $O_{i,j}$ the subset $\{u\in O |\ \exists u_i\in F_i$ and $\exists u_j\in F_j$ such that there exist an $u_i$,$u$-path and an $u_j$,$u$-path in $G[O\cup\{u_i,u_j\}]\}$, for $1\le i <j\le k$. 

%First, we will colour the vertices of $O$ with colours $\{1,\ldots, 7\}$. Second, we will colour the vertices of $\cup_{1\le i\le k} V(F_i)$.

\begin{mathitem} 

\item[{\sl Step 1:}] Colouring the vertices of $O$ with colours $\{1, \ldots,7\}$.

Consider each component $O_j$ of $O$ separately ($i=1,\dots, s$) and let $z_j$ denote the vertex of $O_j$ closest to $F_z$. Then we use an expandable broadcast-colouring to colour vertices of $O_j$ with colours $1,2,\dots, 7$ such that the vertices at distance at most $3$ from $z_j$ are coloured as in Figure~\ref{figadd2} (by considering $z_j$ as $x$ in this figure). Using the result of Sloper~\cite{Slo}, it is possible to extend this packing colouring to the graph $O_j$.
Note that $z_j$ has colour $1$, the neighbour(s) of $z_j$ in $O_j$ has (have) colour $2$ (and $3$), vertices of $O_j$ at distance $2$ from $z_j$ are coloured with colour $1$ and vertices of $O_j$ at distance $3$ have colours $2$, $3$, $4$ and $5$. 
Obviously, since $G$ is subcubic, $z_j$ is at distance at least $3$ from any vertex of any $O_\ell\not=O_i$. 
Note that, except one vertex, every vertex of $O$ having a neighbor in $F_i$ is at shortest distance of $F_z$ (compared to the other vertices in the same component of $O$). Consequently, by definition of $z_j$, in every face $F_i$ there is at most one vertex which has a neighbour in $O$ which is not $z_j$ for some $j\in\{1,\ldots,s\}$.

Let $w_i$ be this possible neighbour in $O$. Note that $w_i$ can have any colour among $\{1,\ldots,7\}$.
If $w_i$ has colour in $\{2,3\}$, then the other vertices of colour $2$ or $3$ at close distance from vertices of $F_i$ are the neighbours of the vertices $z_j$, $j\in\{1,\ldots,s\}$, which are at distance $4$ from $w_i$. 
If $w_i$ has colour in $\{4,5\}$, then, also, the vertices of colour $4$ or $5$ at close distance from vertices of $F_i$ are the vertices at distance $3$ (in $O_j$) of the vertices $z_j$, $j\in\{1,\ldots,s\}$, and these vertices are at distance $6$ from $w_i$. 
Finally, since the vertices at even distance of $z_j$ in $O_j$ are coloured with colour $1$, the other vertices of colour $6$ or $7$ are at distance at least $8$ from $w_i$.
Hence, the above defined colouring satisfies the distance constraints of a packing colouring.

\item[{\sl Step 2:}] Colouring the remaining vertices of $G$.

The sets $V(F_1),\ldots, V(F_r)$ are cycle-distance-preserved sets. Hence, by Lemma~\ref{disjoint cycles}, the remaining uncoloured vertices can be coloured with colours $\{8,\ldots, 9\times 6^{r}-2 \}$.
\end{mathitem}
\end{proof}

\section{Some $2$-connected outerplanar graphs with finite packing chromatic number} 

In this section we consider some special classes of subcubic outerplanar graphs for which we can decrease the upper bound on the packing chromatic number given in Theorem \ref{k-inner face structure}.

\begin{theorem} \label{thm star structure}
If $G$ is a $2$-connected subcubic outerplanar graph with exactly one internal face, then $\chi_{\rho}(G)\leq 51$.
\end{theorem}

\begin{proof}
Suppose $G$ is a $2$-connected subcubic outerplanar graph with exactly one internal face.
Let $C$ denote the boundary cycle of $G$ and $F$ the internal face of $G$.
Let $C'=\{v_0,\ldots, v_{N-1}\}$ denote the set of vertices which belong to $F$, with $v_i$ adjacent to $v_{i+1}$, for $0\le i<N$. When $N$ is odd, we suppose that $v_{N-1}$ is a vertex with $d_G(v_{N-1})=2$. 
Such a vertex exists since the number of vertices of degree 3 in $C'$ is even.
By removing the edges of $C\cap F$ from $G$, and by removing the isolated vertices from the resulting graph, we obtain a graph $G'$ which is a disjoint union of $2$-connected outerplanar graphs having no internal face.

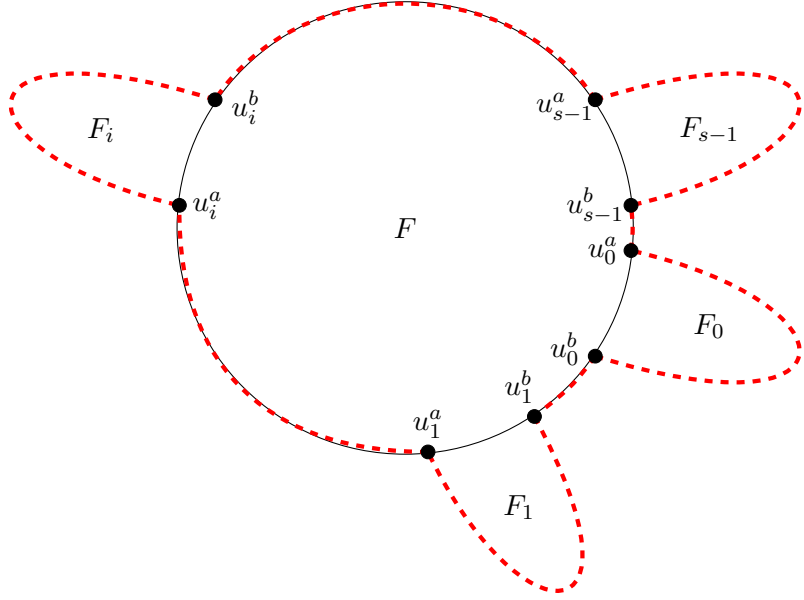
\begin{figure}[t]
\begin{center}
\begin{tikzpicture}
\draw (0,0) circle (3cm);
\draw [ultra thick, style=dashed,color=red]  (2.5,1.7) .. controls (6,2.8) and (6,1) .. (2.97,0.3);
\draw [ultra thick, style=dashed,color=red] (2.5,-1.7) .. controls (6,-2.8) and (6,-1) .. (2.97,-0.3);
\draw [ultra thick, style=dashed,color=red] (1.7,-2.5) .. controls (3.3,-5.5) and (1.5,-5.5) .. (0.3,-2.97);
\draw [ultra thick, style=dashed,color=red] (-2.5,1.7) .. controls (-6,2.8) and (-6,1) .. (-2.97,0.3);
\draw [ultra thick, style=dashed,color=red]  (2.97,-0.3) .. controls (3,0) .. (2.97,0.3);
\draw [ultra thick, style=dashed,color=red]  (2.5,-1.7) .. controls (2.2,-2.1) .. (1.7,-2.5);
\draw [ultra thick, style=dashed,color=red]  (0.3,-2.97) .. controls (-1.7,-2.97) and (-3 ,-1.95) .. (-2.97,0.3);
\draw [ultra thick, style=dashed,color=red]  (2.5,1.7) .. controls (1.4,3.4) and (-1.4,3.4) .. (-2.5,1.7);
\node at (0,0) {$F$};
\node at (4,1.3) {$F_{s-1} $};
\node at (4,-1.3) {$F_0 $};
\node at (1.5,-3.7) {$F_1 $};
\node at (-4,1.3) {$F_i $};
\node at (2.1,-1.6) {$u_0^b$};
\node at (2.5,-1.7) [circle,draw=black,fill=black,scale=0.5]{};
\node at (2.5,0.3) {$u_{s-1}^b$};
\node at (2.97,0.3) [circle,draw=black,fill=black,scale=0.5]{};
\node at (2.6,-0.3) {$u_0^a$};
\node at (2.97,-0.3) [circle,draw=black,fill=black,scale=0.5]{};
\node at (2.1,1.6) {$u_{s-1}^a$};
\node at (2.5,1.7) [circle,draw=black,fill=black,scale=0.5]{};

\node at (1.5,-2.1) {$u_1^b$};
\node at (1.7,-2.5) [circle,draw=black,fill=black,scale=0.5]{};
\node at (0.3,-2.6) {$u_1^a$};
\node at (0.3,-2.97) [circle,draw=black,fill=black,scale=0.5]{};

\node at (-2.1,1.6) {$u_i^b$};
\node at (-2.5,1.7) [circle,draw=black,fill=black,scale=0.5]{};
\node at (-2.6,0.3) {$u_i^a$};
\node at (-2.97,0.3) [circle,draw=black,fill=black,scale=0.5]{};
\end{tikzpicture}
\end{center}
\caption{A $2$-connected outerplanar graph with one internal face and its different subgraphs ($C$ is represented by a dashed line).}
\label{out11}
\end{figure}

Let $F_0$, $\ldots$, $F_{s-1}$ denote the $2$-connected components of $G'$, enumerated in the clockwise order along the cycle $C$ in $G$ (for details, see Fig.~\ref{out11}). Note that, since any $F_i$ contains no internal face, each $F_i$ has exactly two end faces or $F_i$ is a cycle.
Let $i$ be an integer with $0\le i< s$, and let $u_i^a$ and $u_i^b$ denote the two adjacent vertices of degree $3$ in $G$ which belong to $V(F_i)\cap C'$, as it is depicted in Fig.~\ref{out11}.
%%Let $D=\{v_1^a,v_1^b,\ldots v_\ell^a,v_\ell^b\}$.
%%Let $w_i$ be the neighbour of $v_i^a$ in $C$, different of $u_i^a$ and let $y_i$ be a vertex of the other end face of $F_i$.
%%Consider the graph $F_i$, depending the structure of $F_i$, we give the following  colouring to the vertices of $F_i$:
Let $y_i$ be a vertex of degree $2$ in the end face of $F_i$ which does not contain $u_i^a$ (for $F_i$ a cycle we denote by $y_i$ a vertex of $F_i$ at maximum distance from $F$ in $G$).
Let $x_i\in \{u_i^a, u_i^b\}$ denote a vertex at minimal distance from $y_i$. 
Finally, let $P_i$ be a shortest $(x_i,y_i)$-path in $G$. We further denote the vertices of each $P_i$ by $x_i, p_i^1,p_i^2, \ldots, y_i$ in an ordering starting from $x_i$. Let $D_i^1$, $\ldots$, $D_i^{k_i}$ denote the connected components of $F_i-P_i$ with $D_i^1$ containing a vertex among $u_i^a$ and  $u_i^b$ and with $D_i^{k}$ being at larger distance than $D_{i}^{k-1}$ from $x_i$, $2\le k\le k_i$.

The proof will be organized as follows. First, we will colour the vertices of $C'$. Second, we will colour the vertices of $\cup_{0\le i < s} F_i-P_i$ with colour $1$, $2$ and $3$. Note that the obtained colouring does not necessarily satisfy the distance constraints of a packing colouring of $G$. Third, we will modify colouring of some vertices of $F_i-P_i$ ($i=0,\dots, s-1$) to save colour $1$ for some vertices of the paths $P_i$ and to prevent collisions in colour $2$. Fourth, we will recolour some vertices of $F_0,\dots, F_{s-1}$ with colour $4$ in order to satisfy the distance constraints of a packing colouring. Finally, we will colour vertices of the paths $\cup_{0\le i < s} P_i\setminus \{x_i\}$.

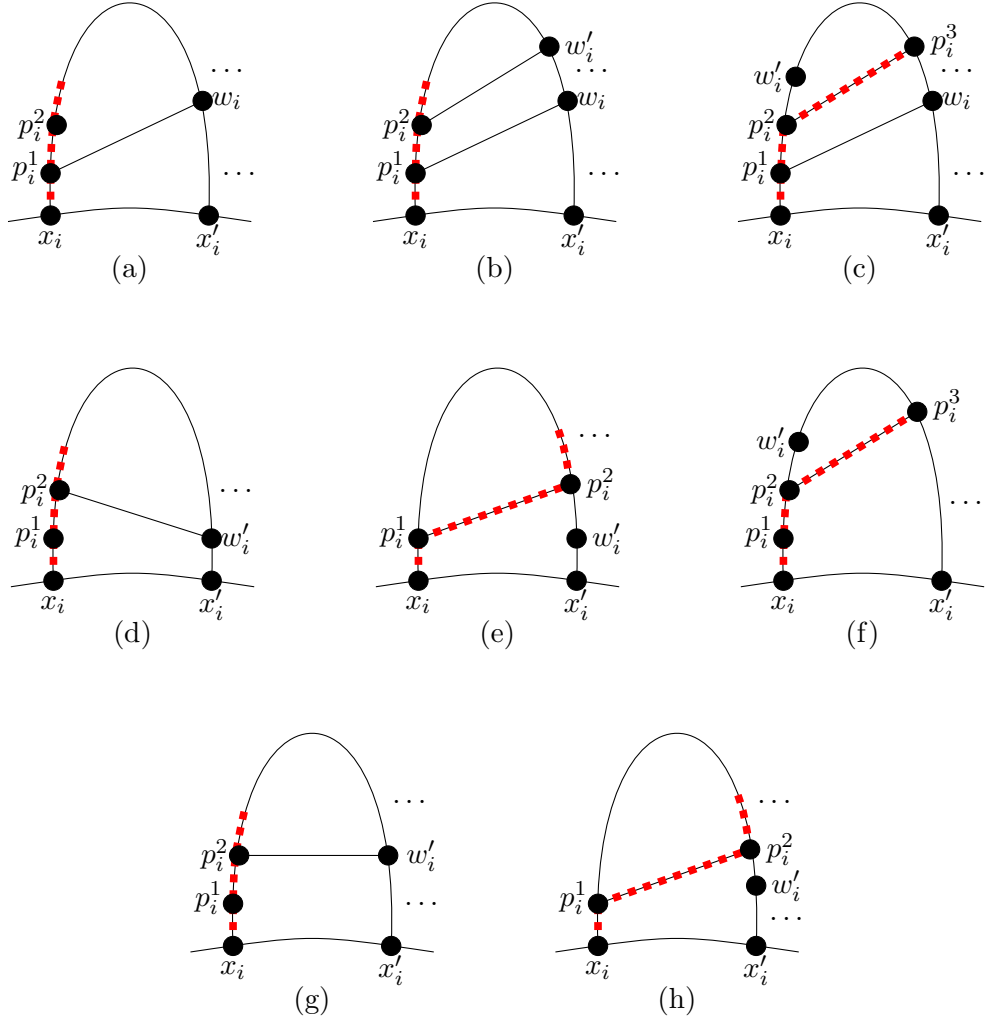
\begin{figure}[t]

\begin{center}

\begin{tikzpicture}[scale=0.8]
\draw  (-2,0) .. controls (0,0.3) .. (2,0);
\draw  (-1.3,0.1) .. controls (-1.5,4.8) and (1.5,4.8) .. (1.3,0.1);
\draw  (-1.3,0.8) -- (1.2,2);
\draw [line width=1mm, style=dashed,color=red] (-1.3,0.1) .. controls (-1.3,1.5) .. (-1.1,2.4);
\node at (-1.3,0.1) [circle,draw=black,fill=black,scale=0.7]{};
\node at (1.3,0.1) [circle,draw=black,fill=black,scale=0.7]{};
\node at (-1.3,0.8) [circle,draw=black,fill=black,scale=0.7]{};
\node at (1.2,2) [circle,draw=black,fill=black,scale=0.7]{};
\node at (-1.2,1.6) [circle,draw=black,fill=black,scale=0.7]{};
\node at (-1.3,-0.3) {$x_i $};
\node at (1.3,-0.3) {$x'_i $};

\node at (-1.6,1.6) {$p_i^2$};
\node at (-1.7,0.9) {$p_i^1$};
\node at (1.6,2) {$w_i $};
\node at (1.6,2.5) {$\ldots$};
\node at (1.8,0.8) {$\ldots$};
\node at (0,-0.8) {(a)};

\draw  (10,0) .. controls (12,0.3) .. (14,0);
\draw  (10.7,0.1) .. controls (10.5,4.8) and (13.5,4.8) .. (13.3,0.1);
\draw  (10.7,0.8) -- (13.2,2);
\draw  (10.8,1.6) -- (12.9,2.9);
\draw [line width=1mm, style=dashed,color=red] (10.7,0.1) .. controls (10.7,0.8) .. (10.75,1.6);
\draw [line width=1mm, style=dashed,color=red] (10.8,1.6) -- (12.9,2.9);
\node at (10.7,0.1) [circle,draw=black,fill=black,scale=0.7]{};
\node at (13.3,0.1) [circle,draw=black,fill=black,scale=0.7]{};
\node at (10.7,0.8) [circle,draw=black,fill=black,scale=0.7]{};
\node at (13.2,2) [circle,draw=black,fill=black,scale=0.7]{};
\node at (10.8,1.6) [circle,draw=black,fill=black,scale=0.7]{};
\node at (10.95,2.4) [circle,draw=black,fill=black,scale=0.7]{};
\node at (12.9,2.9) [circle,draw=black,fill=black,scale=0.7]{};
\node at (10.7,-0.3) {$x_i $};
\node at (13.3,-0.3) {$x'_i $};
\node at (10.5,2.4) {$w'_i $};
\node at (10.4,1.6) {$p_i^2$};
\node at (10.3,0.9) {$p_i^1$};
\node at (13.6,2) {$w_i $};
\node at (13.4,3) {$p_i^3$};
\node at (13.6,2.5) {$\ldots$};
\node at (13.8,0.8) {$\ldots$};
\node at (12,-0.8) {(c)};

\draw  (4,0) .. controls (6,0.3) .. (8,0);
\draw  (4.7,0.1) .. controls (4.5,4.8) and (7.5,4.8) .. (7.3,0.1);
\draw  (4.7,0.8) -- (7.2,2);
\draw  (4.8,1.6) -- (6.9,2.9);
\draw [line width=1mm, style=dashed,color=red] (4.7,0.1) .. controls (4.7,1.5) .. (4.9,2.4);
\node at (4.7,0.1) [circle,draw=black,fill=black,scale=0.7]{};
\node at (7.3,0.1) [circle,draw=black,fill=black,scale=0.7]{};
\node at (4.7,0.8) [circle,draw=black,fill=black,scale=0.7]{};
\node at (4.8,1.6) [circle,draw=black,fill=black,scale=0.7]{};
\node at (7.2,2) [circle,draw=black,fill=black,scale=0.7]{};
\node at (6.9,2.9) [circle,draw=black,fill=black,scale=0.7]{};
\node at (4.7,-0.3) {$x_i $};
\node at (7.3,-0.3) {$x'_i $};
\node at (4.3,0.9) {$p_i^1$};
\node at (4.4,1.6) {$p_i^2$};
\node at (7.6,2.5) {$\ldots$};
\node at (7.6,2) {$w_i $};
\node at (7.4,2.9) {$w'_i $};
\node at (7.8,0.8) {$\ldots$};
\node at (6,-0.8) {(b)};
\end{tikzpicture}
\begin{tikzpicture}[scale=0.8]
\draw  (-2,0) .. controls (0,0.3) .. (2,0);
\draw  (-1.3,0.1) .. controls (-1.5,4.8) and (1.5,4.8) .. (1.3,0.1);
\draw  (-1.2,1.6) -- (1.3,0.8);
\draw [line width=1mm, style=dashed,color=red] (-1.3,0.1) .. controls (-1.3,1.5) .. (-1.1,2.4);
\node at (-1.3,0.1) [circle,draw=black,fill=black,scale=0.7]{};
\node at (1.3,0.1) [circle,draw=black,fill=black,scale=0.7]{};
\node at (-1.3,0.8) [circle,draw=black,fill=black,scale=0.7]{};
\node at (-1.2,1.6) [circle,draw=black,fill=black,scale=0.7]{};
\node at (1.3,0.8) [circle,draw=black,fill=black,scale=0.7]{};
\node at (-1.3,-0.3) {$x_i $};
\node at (1.3,-0.3) {$x'_i $};
\node at (-1.6,1.6) {$p_i^2$};
\node at (-1.7,0.9) {$p_i^1$};
\node at (1.7,0.8) {$w'_i $};
\node at (1.7,1.6) {$\ldots$};
\node at (0,-0.8) {(d)};

\draw  (10-6,0) .. controls (12-6,0.3) .. (14-6,0);
\draw  (10.7-6,0.1) .. controls (10.5-6,4.8) and (13.5-6,4.8) .. (13.3-6,0.1);
\draw  (10.7-6,0.8) -- (13.2-6,1.7);
\draw [line width=1mm, style=dashed,color=red] (10.7-6,0.1) -- (10.7-6,0.8);
\draw [line width=1mm, style=dashed,color=red] (10.7-6,0.8) -- (13.2-6,1.7);
\draw [line width=1mm, style=dashed,color=red] (13-6,2.6) .. controls (13.1-6,2.3) .. (13.2-6,1.7);
\node at (10.7-6,0.1) [circle,draw=black,fill=black,scale=0.7]{};
\node at (13.3-6,0.1) [circle,draw=black,fill=black,scale=0.7]{};
\node at (10.7-6,0.8) [circle,draw=black,fill=black,scale=0.7]{};
\node at (13.2-6,1.7) [circle,draw=black,fill=black,scale=0.7]{};
\node at (13.3-6,0.8) [circle,draw=black,fill=black,scale=0.7]{};
\node at (10.7-6,-0.3) {$x_i $};
\node at (13.3-6,-0.3) {$x'_i $};
\node at (13.7-6,1.7) {$p_i^2 $};
\node at (10.3-6,0.9) {$p_i^1$};
\node at (13.6-6,2.5) {$\ldots$};
\node at (13.8-6,0.8) {$w'_i$};
\node at (12-6,-0.8) {(e)};

\draw  (4+6,0) .. controls (6+6,0.3) .. (8+6,0);
\draw  (4.7+6,0.1) .. controls (4.5+6,4.8) and (7.5+6,4.8) .. (7.3+6,0.1);
\draw  (4.8+6,1.6) -- (6.9+6,2.9);
\draw [line width=1mm, style=dashed,color=red] (4.7+6,0.1) .. controls (4.7+6,0.8) .. (4.75+6,1.6);
\draw [line width=1mm, style=dashed,color=red] (4.8+6,1.6) -- (6.9+6,2.9);
\node at (4.7+6,0.1) [circle,draw=black,fill=black,scale=0.7]{};
\node at (7.3+6,0.1) [circle,draw=black,fill=black,scale=0.7]{};
\node at (4.7+6,0.8) [circle,draw=black,fill=black,scale=0.7]{};
\node at (4.8+6,1.6) [circle,draw=black,fill=black,scale=0.7]{};
\node at (6.9+6,2.9) [circle,draw=black,fill=black,scale=0.7]{};
\node at (4.95+6,2.4) [circle,draw=black,fill=black,scale=0.7]{};
\node at (4.5+6,2.4) {$w'_i $};
\node at (4.7+6,-0.3) {$x_i $};
\node at (7.3+6,-0.3) {$x'_i $};
\node at (4.4+6,1.6) {$p_i^2$};
\node at (4.3+6,0.9) {$p_i^1$};
\node at (7.4+6,3) {$p_i^3$};
\node at (7.7+6,1.4) {$\ldots$};
\node at (6+6,-0.8) {(f)};
\end{tikzpicture}
\begin{tikzpicture}[scale=0.8]
\draw  (-2+6,0) .. controls (0+6,0.3) .. (2+6,0);
\draw  (-1.3+6,0.1) .. controls (-1.5+6,4.8) and (1.5+6,4.8) .. (1.3+6,0.1);
\draw  (-1.2+6,1.6) -- (1.25+6,1.6);
\draw [line width=1mm, style=dashed,color=red] (-1.3+6,0.1) .. controls (-1.3+6,1.5) .. (-1.1+6,2.4);
\node at (-1.3+6,0.1) [circle,draw=black,fill=black,scale=0.7]{};
\node at (1.3+6,0.1) [circle,draw=black,fill=black,scale=0.7]{};
\node at (-1.3+6,0.8) [circle,draw=black,fill=black,scale=0.7]{};
\node at (1.25+6,1.6) [circle,draw=black,fill=black,scale=0.7]{};
\node at (-1.2+6,1.6) [circle,draw=black,fill=black,scale=0.7]{};
\node at (-1.3+6,-0.3) {$x_i $};
\node at (1.3+6,-0.3) {$x'_i $};
\node at (-1.6+6,1.6) {$p_i^2$};
\node at (-1.7+6,0.9) {$p_i^1$};
\node at (1.8+6,1.6) {$w'_i $};
\node at (1.6+6,2.5) {$\ldots$};
\node at (1.8+6,0.8) {$\ldots$};
\node at (0+6,-0.8) {(g)};

\draw  (10,0) .. controls (12,0.3) .. (14,0);
\draw  (10.7,0.1) .. controls (10.5,4.8) and (13.5,4.8) .. (13.3,0.1);
\draw  (10.7,0.8) -- (13.2,1.7);
\draw [line width=1mm, style=dashed,color=red] (10.7,0.1) -- (10.7,0.8);
\draw [line width=1mm, style=dashed,color=red] (10.7,0.8) -- (13.2,1.7);
\draw [line width=1mm, style=dashed,color=red] (13,2.6) .. controls (13.1,2.3) .. (13.2,1.7);
\node at (10.7,0.1) [circle,draw=black,fill=black,scale=0.7]{};
\node at (13.3,0.1) [circle,draw=black,fill=black,scale=0.7]{};
\node at (10.7,0.8) [circle,draw=black,fill=black,scale=0.7]{};
\node at (13.2,1.7) [circle,draw=black,fill=black,scale=0.7]{};
\node at (13.3,1.1) [circle,draw=black,fill=black,scale=0.7]{};
\node at (10.7,-0.3) {$x_i $};
\node at (13.3,-0.3) {$x'_i $};
\node at (13.7,1.7) {$p_i^2 $};
\node at (10.3,0.9) {$p_i^1$};
\node at (13.6,2.5) {$\ldots$};
\node at (13.8,1.1) {$w'_i$};
\node at (13.8,0.55) {$\ldots$};
\node at (12,-0.8) {(h)};
\end{tikzpicture}
\end{center}
\caption{Eight configurations in step 3 ($P_i$ is represented by a dashed line).}
\label{out12}
\end{figure}
\begin{mathitem}

\item[{\sl Step 1:}] Colouring the vertices of $C'$ with colours $1$, $2$, $29, 30,\dots, 45$.

For integers $j$, $j'$, let $r_{j,j'}$ be an integer such that $r_{j,j'}\equiv j-j' \pmod{N}$ and $- \lfloor N/2 \rfloor\le r_{j,j'}\le \lfloor N/2 \rfloor$.
Note that $d_G(v_j,v_{j'}) =|r_{j,j'}|$. We begin with a partitioning of $C'$ into five subsets: $C'_1=\{v_j|\ j\equiv 0\pmod{2},\  0\le j< N\}$, $C'_2=\{v_j|\ j\equiv 1\pmod{4},\  0\le j< N\}$, $C'_3=\{v_j|\ j\equiv 3\pmod{12},\  0\le j< N\}$, $C'_4=\{v_j|\ j\equiv 7\pmod{12},\  0\le j< N\}$ and $C'_5=\{v_j|\ j\equiv 11\pmod{12},\  0\le j< N\}$.
Let $m_k$ denote the vertex with the largest index in $C'_k$, for $k\in\{1,2,3,4,5\}$.
We use the following patterns to colour the vertices of $C'$.
\begin{enumerate}
\item if $|C'_1|\equiv 0\pmod{2}$ (or $|C'_1|\equiv 1\pmod{2}$), then we colour all vertices of $C'_1$ (or $C'_1\setminus \{m_1\}$, respectively) with colour $1$;
\item we colour all vertices of $C'_2\setminus \{m_2\}$ with colour $2$;
\item if $|C'_3|\equiv 0\pmod{4}$ (or $|C'_3|\equiv 1\pmod{4}$, respectively), then we use the pattern $29,30,35,36,$ $29,30,35,36,\ldots, 29,30,35,36$ to colour the vertices of $C'_3$ (or $C'_3\setminus \{m_3\}$, respectively);

if $|C'_3|\equiv 2\pmod{4}$ (or $|C'_3|\equiv 3\pmod{4}$), then we use the pattern $29,30,35,$$36,$ $29,30,$$35,36,\ldots,29,30,35,$ $29,30,36$  to colour the vertices of $C'_3$ (or $C'_3\setminus \{ m_3\}$, respectively);
\item if $|C'_4|\equiv 0\pmod{4}$ (or $|C'_4|\equiv 1\pmod{4}$), then we use the pattern $31,32,37,$ $38,$ $31,32,$ $37,38,\ldots,31,32,37,38$ to colour the vertices of $C'_4$ (or $C'_4\setminus \{ m_4\}$, respectively);

if $|C'_4|\equiv 2\pmod{4}$ (or $|C'_4|\equiv 3\pmod{4}$), then we use the pattern $31,32,37,$ $38,31,$ $32,$ $37,38,\ldots,31,$ $32,37,31,32,38$  to colour the vertices of $C'_4$ (or $C'_4\setminus \{m_4\}$, respectively);
\item if $|C'_5|\equiv 0\pmod{4}$ (or $|C'_5|\equiv 1\pmod{4}$), then we use the pattern $33,34,39,$ $40,$ $33,34,$ $39,40,\ldots,33,$ $34,39,40$ to colour the vertices of $C'_5$ (or $C'_5\setminus \{m_5\}$, respectively);

if $|C'_5|\equiv 2\pmod{4}$ (or $|C'_5|\equiv 3\pmod{4}$), then we use the pattern $33,34,39,$ $40,$ $33,34,$ $39,40,\ldots,33,$ $34,39,$ $33,34,40$ to colour the vertices of $C'_5$ (or $C'_5\setminus \{ m_5\}$, respectively);
\item when it is necessary, we use the colours $41$, $42$, $43$, $44$, $45$ to colour the vertices of $\{m_k |\ 1\le k\le 5\}$.
\end{enumerate}
One can check that, for any pair of vertices $u,v$ of $C'$ with the same colour $k$, $d_G(u,v)>k$.
For example, in the pattern $29,30,35,36$ of length four, two vertices with the same colour are at distance at least $48$, since the pattern has length four and we colour vertices with the same remainder modulo $12$. The same goes for the pattern $29,30,35,29,30,36$ of length six. Note that, for every pair of vertices $(u_i^a,u_i^b)$, at least one of them is coloured with $1$.

\item[{\sl Step 2:}] Colouring the vertices of $F_i- P_i$ with colours $1$, $2$ and $3$, for every $i=0,\ldots, s-1$.

Let $x'_i$ be the vertex among $u_i^a$ and $u_i^b$ different from $x_i$. Let  $l_i$ be the order of $D_i^1$ and let $x'_i, x_i^1, \ldots, x_{i}^{ l_{i}-1}$ be the vertices of $D_i^1$  in an ordering starting from $x'_i$.
If $x'_i$ is coloured with colour $1$, then we use the pattern $3,1,2,1$ to colour the vertices $x_i^1, \ldots, x_{i}^{l_{i}-1}$. If $x'_i$ is not coloured with colour $1$, then we use the pattern $1,3,1,2$ to colour the vertices $x_i^1, \ldots, x_{i}^{l_i-1}$. 
Analogously as in the proof of Lemma \ref{lemma}, we colour vertices of $D_i^j$ ($j=2,3,\dots, k_i$) using the pattern $1,2,1,3$ starting from the vertex of $D_i^j$ at shortest distance from $C'$.
At this step we do not change the colouring in order to avoid clashing vertices of colour $3$.
%If there are collisions between two vertices with colour $3$, we recolour vertices as in the proof of Lemma \ref{lemma}.

\item[{\sl Step 3:}] Recolouring some vertices in $F_i$, for every $i=0,\ldots, s-1$.

In this step we deal with possible collisions in colour $2$ between vertices of $F$ and vertices of $D_i^1$ at distance $2$ from $F$. We also change colours of neighbours of $p_i^2$ coloured with $1$ since, in Step $5$, we will colour $p_i^2$ with $1$ for reducing the number of colours used for the whole graph $G$. 

Since we used the patterns $1,3,1,2$ and $3,1,2,1$ to colour the vertices of $V(D_i^1)\setminus \{x_i\}$, no vertex at distance $2$ from $x'_i$ has colour $2$. For any $i=0,\dots, s-1$, let $w_i$ denote the vertex of $F_i-P_i$ at distance $2$ from $x_i$ and let $w'_i$ be the possible neighbour of $p_i^2$ in $F_i-P_i$.

%We consider four cases depending on whether there exists a vertex $w_i$ of colour $2$ which is at distance $2$ from $x_i$ or not and depending  on whether there exits a neighbour $w'_i$ of colour $1$ of $p_i^2$ or not.
%We can also note that $w_i$ belongs to the chord delimiting the end face of $F_i$ containing $x_i$ and that $p_i^2$ can have at most one neighbour of colour $1$, even if $p_i^2$ is $y_i$ (since $y_i$ has degree $2$).

\begin{enumerate}
\item[Case i)]{\sl $w_i$ has colour $2$.} First suppose that $p_i^2$ has no neighbour of colour $1$ (see Fig. \ref{out12}(a)) or $p_i^2$ has a neighbour $w'_i$ with colour $1$ in $D_i^1$ (see Fig. \ref{out12}(b)). In both possibilities we recolour the vertices $x_i^j=w_i, x_i^{j+1}, x_i^{l_i-1}$ of $D_i^1$ with $\underline{4},2,1,3,1,2,1\ldots$ instead of $\underline{2},1,3,1,2,\ldots$. Note that the underlined colours belong to the vertex $w_i$.

Now we assume that $p_i^2$ has a neighbour $w'_i$ of colour $1$ which does not belong to $D_i^1$. Thus $p_i^2 p_i^3$ is a chord (see Fig. \ref{out12}(c)). We recolour the vertices $x_i^j=w_i, x_i^{j+1}, \dots, x_i^{l_i-1}$ of $D_i^1$  with $\underline{4},1,2,1,3,\ldots$ instead of $\underline{2},1,3,1,2,\ldots$ (the underlined colours belong to $w_i$), and the vertices of $D_i^2$ with pattern $\underline{2},1,3,1,\ldots$ instead of $\underline{1},2,1,3,\ldots$ (the underlined colours belong to $w'_i$). Note that if $w'_i$ had colour $1$, the colour $1$ was changed.

\item[Case ii)] {\sl $w_i$ does not have colour $2$}. If $p_i^2$ has no neighbour $w'_i$ of colour $1$, then we do not modify the colouring of $D_i^1$ in this step. Suppose that $p_i^2$ has a neighbour $w'_i$ of colour $1$. Suppose that $w'_i$ is a neighbour of $x'_i$ (see Fig. \ref{out12}(d,e)). Clearly $x'_i$ does not have colour $1$ and $x_i$ has colour $1$ (by Step 1). Then we modify the path $P_i$ by replacing vertex $x_i$ with $x'_i$ and $p_i^1$ with $w'_i$ and recolour the modified path $D^i_1$ with pattern $3,1,2,1,3,\dots $.

Now suppose that $w'_i$ is at distance at least $2$ from $x'_i$ and that $p_i^2 w'_i$ and $p_i^1 p_i^2$ are not chords (see Fig. \ref{out12}(f)). We recolour vertices of $D_2^i$ with pattern $\underline{2},1,3,1,\ldots$ instead of $\underline{1},2,1,3,\ldots$. If $w'_i$ is at distance at least $2$ from $x'_i$ and $p_i^2 w'_i$ is a chord (see Fig. \ref{out12}(g)), we recolour vertices $x_i^j=w'_i,x_i^{j+1}, \dots x_i^{l_i-1}$ of $D_1^i$ with $\underline{4},1,2,1,3,\ldots$ instead of $\underline{1},2,1,3,\ldots$ or $\underline{1},3,1,2,\ldots$.
Finally, if $w'_i$ is at distance at least $2$ from $x'_i$ and $p_i^1 p_i^2$ is a chord (see Fig. \ref{out12}(h)), then we change the colour of $w'_i$ to $4$. Note again that, in each possibility, the underlined colours belong to $w'_i$.
\end{enumerate}

\item[{\sl Step 4:}] Avoid collisions in colouring of vertices of $F_i-P_i$ for $i=0,\dots,s-1 $. 

Now we check and modify (analogously as in the proof of Lemma \ref{lemma}) the defined colouring of $F_i-P_i$ to avoid collisions between pairs of vertice with the same colour. Obviously, there is no collision between vertices coloured with colour $1$ or $2$. Hence the only possible collision is in colour $3$. Let $a$ and $b$ be a pair of clashing vertices in colour $3$. If $a$ and $b$ belong to the same component $D_i^k$ of $F_i - P_i$, $k,\in\{1,2,\ldots,k_i\}$, then we proceed as in Case 2 of the proof of Lemma \ref{lemma}.
Thus we may assume that $a$ and $b$ belong to different components $D_i^k$ and $D_i^{k'}$ of $F_i - P_i$, $k,k'\in\{1,2,\ldots,k_i\}$. If we changed the colouring of $D_i^2$ in Step 3, then we recolour the vertices of $D_i^2$ starting from $w'_i$ (also defined in Step 3) with pattern $2,3,1,2,1,3,1,\dots $ instead of $2,1,3,1,2,1,\dots$. Then we proceed as in Case 1 of the proof of Lemma \ref{lemma}. %Note that we do not change the colouring of the vertices of $D_i^1$.

\medskip
 
Now we have to check that the vertices coloured with colour $4$ in Step 3 are pairwise at distance at least $5$, and that the vertices coloured with colour $4$ in Step 3 are pairwise at distance at least $5$ from the added vertices of colour $4$ in Step $4$ of applying Lemma \ref{lemma}.

By definition, $w_i$ is at distance at most $2$ from both $x_i$ and $x'_i$. Moreover, since $w'_i$ is the neighbour of $p_i^2$ and since we did not recolour $w'_i$ with colour $4$ in the configurations described in Fig. \ref{out12}(d,e) of Step 3, $w'_i$ is either at distance at least $2$ from both $x_i$ and $x'_i$ or does not have colour $4$.
Thus, the vertices recoloured with colour $4$ in Step 3 are at mutual distance at least $5$.

Let $a$ be a vertex of colour $4$ from Step 3 (one of $w_i$, $w'_i$ denoted in Step 3). Suppose that $b$ is a vertex of colour $4$ not in $D_i^1$. The minimal distance between any vertex of $D_i^1$ and any vertex of $D_i^2$ is at least $3$. Moreover, because we have proceeded as in the proof of Lemma \ref{lemma}, $b$ is at distance at least $2$ from a vertex at minimal distance from $a$. Hence, $d(a,b)\ge 5$. 

Now suppose $b$ is a vertex of colour $4$ in $D_i^1$. Since, in every case, $a$ is at distance at least $3$ from another vertex of colour $3$ in $D_i^1$, we obtain that $d(a,b)\ge 5$.

\item[{\sl Step 5:}] Colouring the vertices of $P_i\setminus\{x_i\}$ with colours 5 to 28 and 46 to 51, for every $i=0,\ldots, s-1$.

We start with colouring of the vertices $p^2_i$ by $1$ for each $i=0,\dots s-1$. Since we have changed the colours of the eventual neighbours of $p^2_i$ of colour $1$ in Step 3, there are no possible collisions.

For the vertices of $P_i$, we use Pattern~\eqref{pat5-15} beginning at the fourth vertex of $P_i$, i.e., the vertex $p^3_i$. By the proof of Lemma \ref{out2}, we know that such a colouring satisfies the distance constraints of a packing colouring.
 
Let $B=\{p^1_{i}|\ 0\le i< s \} $. We colour the vertices of $B$ with colours $16$ to $28$ and (if necessary) $46$ to $51$.
For integers $j$, $j'$, let $r_{j,j'}$ be an integer such that $r_{j,j'}\equiv j-j' \pmod{s}$ and $- \lfloor s/2 \rfloor\le r_{j,j'}\le \lfloor s /2 \rfloor$.
Note that the vertices $p^1_j$ and $p^1_{j'}$ are at distance $2 \vert r_{j,j'}\vert +1$.
We begin by a partitioning of $B$ into three subsets $B_1$, $B_2$ and $B_3$, with $B_k=\{p_i^1|\  i\equiv k-1\pmod{3},\  0\le i< s \}$, $k=1,2,3$.
Let $m_k$ ($m'_k$) denote the vertex with the largest (second largest, respectively) index in $B_k$, for $k\in\{1,2,3\}$.
We use the following patterns to colour the vertices of $B$.
\begin{enumerate}
\item For vertices of $B_1$, we use the pattern $16, 17, 18, 16, 17, 18, \ldots, 16, 17, 18$. If $|B_1|\equiv 1\pmod{3}$ (or $|B_1|\equiv 2\pmod{3}$), then we erase colours of $m_1$ (or of $m_1$ and $m'_1$, respectively).
\item For vertices of $B_2$, we use the pattern $$19, 20, 21, 25, 26, 19, 20, 21, 25, 26, \ldots, 19, 20, 21, 25, 26$$ when $\vert B_2\vert \equiv 0,1,2 \pmod{5}$, or the pattern $$19, 20, 21, 25, 26, 19, 20, 21, 25, 26, \ldots, 19, 20, 21, 25, 26, 19, 20, 21, 25, 19, 20, 21, 26$$ when $\vert B_2\vert \equiv 3,4 \pmod{5}$. Then, for $\vert B_2\vert \equiv 1,4 \pmod{5}$ we erase colour of $m_2$, and for $\vert B_2\vert \equiv 2 \pmod{5}$ we erase colours of $m_2$ and $m'_2$.

\item For vertices of $B_3$, we use the pattern $$22,23,24,27,28,22,23,24,27,28,\ldots,22,23,24,27,28$$ when $\vert B_3\vert \equiv 0,1,2 \pmod{5}$, or the pattern $$27,28,22,23,24,27,28,\ldots,22,23,24,27,28,22,23,24,27,22,23,24,28$$ when $\vert B_3\vert \equiv 3,4 \pmod{5}$. Then, for $\vert B_3\vert \equiv 1,4 \pmod{5}$ we erase colour of $m_3$, and for $\vert B_3\vert \equiv 2 \pmod{5}$ we erase colours of $m_3$ and $m'_3$.

\item When it is necessary, we use the colours $46$, $47$, $48$, $49$, $50$, $51$ to colour the vertices of $\{m_k,m'_k |\ 1\le k\le 3\}$.
\end{enumerate}
\end{mathitem}

For checking that the defined colouring satisfies the distance constraints of a packing colouring, we recall that any two consecutive vertices in each $B_k$ ($k=1,2,3$) are pairwise at distance at least $7$, implying that vertices having the same colour are pairwise at distance at least $19$ in $B_1$ and at distance at least $31$ in $B_2$ and in $B_3$ (except some vertices of colours from $19$ to $24$ that can be at distance $25$ apart).
\end{proof}

Note that in some cases (depending on the size of $B$ and $C'$) we can decrease the upper bound $51$ of Theorem \ref{thm star structure}. For example, if the internal face $C'$ has length $4k$ ($k\in\mathbb{N}$) and if the number of $2$-connected components of $G-C'$ is $15r$, $r\in \mathbb{N}$, then we can use only $40$ colours instead of $51$.

The following statement will be used in the proof of Theorem \ref{thm path structure}.

\begin{proposition} \label{pathgap}
There is a packing colouring of even vertices of $P_{\infty}$ with colours $\{k,k+1, \dots, 2k-1\}$.
\end{proposition}

\begin{proof}
For the colouring of the vertices of the path $P_{\infty}$ we use pattern $1,k,1,k+1,1,k+2, \dots, 1,2k-1$ and after deleting colour $1$ we get the required colouring. Note that the distance between any pair of vertices coloured with the same colour in two consecutive copies of this pattern is $2k$.
\end{proof}

The last class of outerplanar graphs we consider in this paper is a class of not necessarily $2$-connected graphs.

\begin{theorem} \label{thm path structure}
Let $G$ be a connected subcubic outerplanar graph with no internal face such that the block graph of $G$ is a path. Then $\chi_{\rho}(G)\leq 305$.
\end{theorem}

\begin{proof}

Let $G$ be a sucubic outerplanar graph, $B_G$ the block graph of $G$ and let $B_1, \dots, \, B_k$ denote the blocks of $G$ such that $B_i, B_{i+1}$ are consecutive in $B_G$, $i=1,\dots, k-1$ (i.e., since $\Delta(G)\le 3$, they are connected by a path which intersects no other block of $G$). Let $C_i$ denote the boundary cycle of $B_i$, $i=1,\dots, k$. Since $G$ contains no internal face, each $B_i$ contains no internal face as well, implying that every $B_i$ which is not a cycle contains exactly two end faces. Let $x_1$ denote any vertex of degree two (in $B_1$) in one end face of $B_1$, $x_k$ any vertex of degree two (in $B_k$) in one end face of $B_k$ and let $P$ denote a shortest $x_1,x_k$-path in $G$. Among all possible choices of the vertices $x_1,x_k$ we choose $x_1'$ and $x_k'$ such that the path $P$ is shortest possible. Let $P_i=P\cap B_i$, $i=1,\dots, k$. Obviously $P$ goes through all the blocks of $G$, hence $P_i$ is nonempty for each $i=1,\dots, k$, every $P_i$ is a path since the block graph of $G$ is a path, and each $P_i$ is shortest in $G$ since $P$ is shortest in $G$. In an orientation of $P$ from $x_1$ to $x_k$, we denote by $z_i$ the first vertex of $P$ in $B_i$ and by $z_i'$ the neighbour of $z_i$ in $B_i$ which does not belong to $P$. Note that there must be exactly one such vertex $z_i'$ in each $B_i$ since $P$ is shortest possible and, clearly, each $z_i$ must have degree three. For every block which is not a cycle, we further denote by $x_i$ any vertex of degree two (in $B_i$) in one end face of $B_i$, $i=2,\dots, k-1$ and by $y_i$ any vertex of degree two (in $B_i$) in the end face of $B_i$ which does not contain vertex $x_i$, $i=1,\dots, k$. 
\begin{figure}[ht]
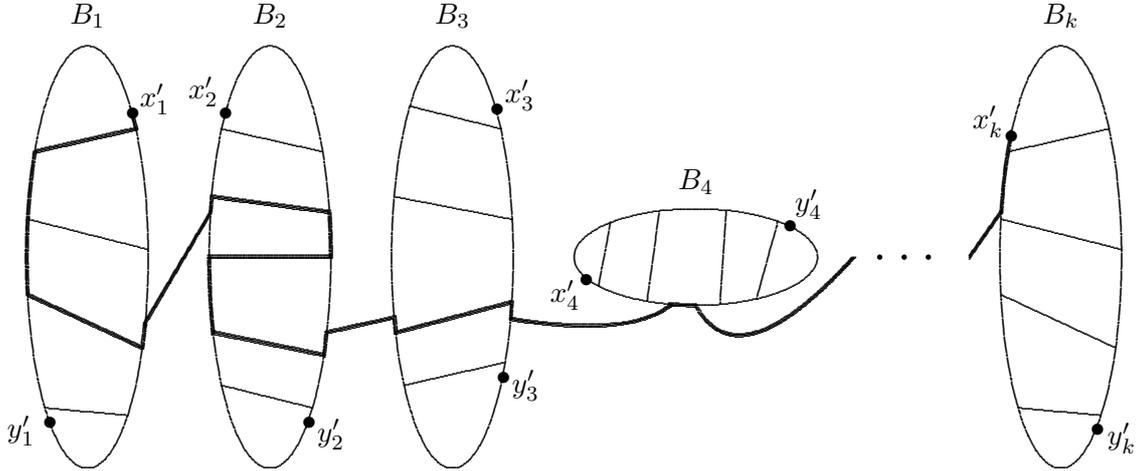

\beginpicture
\setcoordinatesystem units <0.8mm,1mm>
\setplotarea x from -10 to 180, y from 20 to 90
\setlinear
\ellipticalarc axes ratio 2:7 360 degrees from 20 50 center at 10 50
\put{$B_1$} at 10 82
\plot 1.5 64  18 67 /
\plot 0.3 55  20 51 /
\plot 0.2 45  18.9 38 /
\plot 3 30  16.6 29 /
\put{$x'_1$} at 21 71
\put{$y'_1$} at -1 27
\put{$\bullet$} at 17.4 69
\put{$\bullet$} at 3.8 28

\ellipticalarc axes ratio 2:7 360 degrees from 50 50 center at 40 50
\put{$B_2$} at 40 82
\plot 32 67  48.5 64 /
\plot 30.5 58  49.5 56 /
\plot 30 50  50 50 / 
\plot 30.7 40  48.9 37 /
\plot 32 33  47 30 /
\put{$x'_2$} at 29 72
\put{$y'_2$} at 50 26.5
\put{$\bullet$} at 32.7 69
\put{$\bullet$} at 46.4 28

\ellipticalarc axes ratio 2:7 360 degrees from 80 50 center at 70 50
\put{$B_3$} at 70 82
\plot 63 70  78 67 /
\plot 60.5 58  79.8 55 /
\plot 60.7 40  79.6 44 /
\plot 62.1 33  78.5 36 /  
\put{$x'_3$} at 81 71.5
\put{$y'_3$} at 82 33
\put{$\bullet$} at 77.3 69.5
\put{$\bullet$} at 78.4 34

\ellipticalarc axes ratio 5:2 360 degrees from 90 50 center at 110 50
\put{$B_4$} at 110 60
\plot 94 46.1  96 54.6 /
\plot 102 44.1  104 56 /
\plot 114 43.8  115 56.1 /
\plot 120 44.5  123.5  54.7 /
\put{$x'_4$} at 88.5 45
\put{$y'_4$} at 128.5 57
\put{$\bullet$} at 92 46.9
\put{$\bullet$} at 125.5 54

\huge
\put{$\dots$} at 145 50
\normalsize

\ellipticalarc axes ratio 2:7 360 degrees from 180 50 center at 170 50
\put{$B_k$} at 170 82
\plot 161.5 64  178 67 /
\plot 160.3 55  180 51 /
\plot 160.2 45  178.9 38 /
\plot 163 30  176.6 29 /
\put{$x'_k$} at 158 68
\put{$y'_k$} at 180 26
\put{$\bullet$} at 161.9 66
\put{$\bullet$} at 176 27

%   path P
\plot 18 67.2   1.5 64.2  1.5 63.8  18.2 66.85  18.2 67  17.6 69  17.2 69  17.8 67 /
\ellipticalarc axes ratio 2:7 40 degrees from 1.17 64 center at 10 50
\ellipticalarc axes ratio 2:7 41 degrees from 1.6 64 center at 10 50
\plot 0.2 45.2  18.9 38.2  18.9 37.8  0.2 44.8 / 
\plot 18.8 38  19.25 41  19.65 41  19.2 37.8 /
\plot 19.35  40.75  30.3 55.75  30.3 56.25  19.25 41.25  19.25  41  30.3 56 /
\plot 30.05 56  30.25 58.2  49.5 56.2  49.5 55.8  30.65  57.8  30.45 56 / 
\ellipticalarc axes ratio 2:7 12.5 degrees from 49.8 50 center at 40 50
\ellipticalarc axes ratio 2:7 12.8 degrees from 50.2 49.8 center at 40 50
\plot 30 50.2  50 50.2  50 49.8  30 49.8 /  
\ellipticalarc axes ratio 2:7 21.7 degrees from 29.8 50.2 center at 40 50
\ellipticalarc axes ratio 2:7 21.8 degrees from 30.2 50 center at 40 50
\plot 30.5 39.8  48.9 36.8  48.9 37.2  30.7 40.2  30.7 40 /
\plot 48.7 37  49.1 40  49.5 40  49.1 36.8 /
\plot 49.3 39.8  60.25 41.8  60.5 39.75  79.6 43.8  79.6 44.2  60.9 40.2  60.65 42.2  49.3 40.2  49.3 40  60.45 42 /
\plot 79.6 44  79.4 41.8  79.8 42  80 44.2 /
\setquadratic
\plot 79.6 41.8  96 41  106 43.7 /
\plot 79.6 41.96  96 41.17  105.7 43.75 /
\plot 79.4 41.64  96 40.83  106.3 43.65 /
\setlinear
\plot 106 43.52  110 43.42  110 43.78  105.6 43.88  105.8 43.7  110 43.6 /
\setquadratic
\plot 110 43.6  120 40  136 50  /
\plot 109.75 43.6  120 39.82  136.25 50  /
\plot 110.25 43.6  120 40.18  135.75 50  /
\setlinear
\plot 155 49.75  160.2 55.75  160.2 56.25  155 50.25  155 50  160.2 56 /
\ellipticalarc axes ratio 2:7 -24 degrees from 160.05 56 center at 170 50
\ellipticalarc axes ratio 2:7 -24 degrees from 160.4 55.8 center at 170 50

\endpicture
\caption{Structure of blocks of the graph $G$ in Theorem \ref{thm path structure}.}
\label{fig path structure}
\end{figure}

\medskip

Among all possible choices of the vertices $x_i,y_i$ ($i=1,\dots, k$) we choose $x_i', y_i'$ in such a way that $x_i',P$-path and $y_i',P$-path, respectively, is shortest possible. For each $i=1,\dots, \, k$, let $Q_{2i-1}$ be a shortest $y'_i,P$-path in $G$, and let $Q_{2i}$ be a shortest $x'_i,P$-path in $G$. Note that $Q_2$ and $Q_{2k}$ are empty and that some other paths $Q_{i}$ may be trivial or empty (e.g., in the case when $B_i$ is a cycle). The structure of the graph $G$ is depicted in Fig. \ref{fig path structure} (the thick path represents the path $P$) and the paths $Q_{2i-1}, Q_{2i}$ are illustrated in Fig.~\ref{figadd44}.

Consider each block $B_i$ separately. Note that, for each $i=1,\dots, \,k$, the graph $G_i=G[B_i-(V(P_i)\cup V(Q_{2i-1})\cup V(Q_{2i}))]$ consists of path components. Thus, analogously as in the proof of Lemma \ref{lemma}, we can colour vertices of $G_i$ with colours $1,2,3,4$ using the periodic pattern $1,2,1,3$ and modifications introduced in the proof of Lemma \ref{lemma}, starting at the end face of $B_i$ containing vertex $x_i$. Note that, if $B_i$ is a cycle, $\chi_{\rho}(B_i)\leq 4$. Moreover, we can colour the vertices of $G_i$ in such a way that the vertex $z_i'$ gets colour $1$ (if not so, then we can interchange roles of $x_i$ and $y_i$ for colouring of $V(G_i)$, i.e., we start such a colouring from $y_i$ instead of $x_i$).

Now we check that the defined colouring meets the conditions of a packing colouring. First, there is no collision in colours $1$ and $2$, since $\mbox{dist}_G(a,b)\geq 3$ for any $a\in V(G_i)$ and $b\in V(G_j), \, i<j$.
From the modifications described in the proof of Lemma \ref{lemma}, it follows that no end vertex of any path component of any $G_i$ is coloured with colour $4$ since colour $4$ was used for a vertex between vertices of a critical pair belonging to one path component.
This implies that the distance between two vertices coloured with $4$ which belong to different block $B_i,B_j$ is at least $5$, hence there is no collision in colour $4$.
For colour $3$, since no $z_i'$ is coloured with colour $3$, there is no collision in colour $3$ as well.
And since there is no edge in $G$ connecting the path components of the blocks $B_i$, the defined colouring meets the distance constraints of a packing colouring of $\bigcup\limits_{i=1}^k \left[ B_i-(V(P_i)\cup V(Q_{2i-1})\cup V(Q_{2i}))\right]$.

\begin{figure}
\begin{center}
\begin{tikzpicture}[scale=1]
    \draw (0,0) ellipse (6cm and 1cm);
    \draw (0,3) ellipse (6cm and 1cm);
   \draw[ultra thick,color=red] (-1,1)-- (-1,2);
   \draw[ultra thick,color=red] (-1,2)-- (0.5,2);
   \draw[ultra thick,color=red] (1,4)-- (0.5,2);
   \draw[ultra thick,color=red] (1,4)-- (1.5,3.95);
   \draw[ultra thick,color=red] (2,3.9)-- (1.5,3.95);

   \draw[ultra thick,color=red] (5,-0.55)-- (5.5,-1.1);
   \draw[ultra thick,color=red] (2,3.9)-- (2.5,4.4);

   \draw[ultra thick,color=red] (-1,1)-- (-0.5,1);
   \draw[ultra thick,color=red] (-0.5,1)-- (0,-1);
   \draw[ultra thick,color=red] (1.5,-1)-- (0,-1);
   \draw[ultra thick,color=red] (1.5,-1)-- (2,-0.95);
   \draw[ultra thick,color=red] (2.5,-0.9)-- (2,-0.95);
   \draw[ultra thick,color=red] (2.5,-0.9)-- (3,-0.85);
   \draw[ultra thick,color=red] (3.5,-0.8)-- (3,-0.85);
   \draw[ultra thick,color=red] (3.5,-0.8)-- (4,-0.75);
   \draw[ultra thick,color=red] (4.5,-0.65)-- (4,-0.75);
   \draw[ultra thick,color=red] (4.5,-0.65)-- (5,-0.55);
   \draw[ultra thick,color=blue,dashed] (5.5,-0.4)-- (5,-0.55);
   \draw[ultra thick,color=blue,dashed] (5.5,-0.4)-- (6,0);

   \draw[ultra thick,color=blue,dashed] (5.5,3.4)-- (5,3.55);
   \draw[ultra thick,color=blue,dashed] (5,3.55)-- (4.5,3.65);
   \draw[ultra thick,color=blue,dashed] (4.5,3.65)-- (4,3.75);
   \draw[ultra thick,color=blue,dashed] (3.5,3.8)-- (4,3.75);
   \draw[ultra thick,color=blue,dashed] (3.5,3.8)-- (3,3.85);
   \draw[ultra thick,color=blue,dashed] (2.5,3.9)-- (3,3.85);
   \draw[ultra thick,color=blue,dashed] (2.5,3.9)-- (2,3.95);

   \draw[ultra thick,color=blue,dashed] (-1.5,2)-- (-1,2);
   \draw[ultra thick,color=blue,dashed] (-1.5,2)-- (-2,2.05);
   \draw[ultra thick,color=blue,dashed] (-2.5,2.1)-- (-2,2.05);
   \draw[ultra thick,color=blue,dashed] (-2.5,2.1)--  (-3,3.85);
   \draw[ultra thick,color=blue,dashed] (-3.5,3.8)--  (-3,3.85);
   \draw[ultra thick,color=blue,dashed] (-3.5,3.8)--  (-4,3.75);
   \draw[ultra thick,color=blue,dashed] (-4.5,3.65)--  (-4,3.75);
   \draw[ultra thick,color=blue,dashed] (-4.5,3.65)--  (-5,3.55);

   \draw[ultra thick,color=blue,dashed] (-1.5,1)-- (-1,1);
   \draw[ultra thick,color=blue,dashed] (-1.5,1)-- (-2,0.95);
   \draw[ultra thick,color=blue,dashed] (-2.5,0.9)-- (-2,0.95);
   \draw[ultra thick,color=blue,dashed] (-2.5,0.9)--  (-3,0.85);
   \draw[ultra thick,color=blue,dashed] (-3.5,0.8)--  (-3,0.85);
   \draw[ultra thick,color=blue,dashed] (-3.5,0.8)--  (-4,0.75);
   \draw[ultra thick,color=blue,dashed] (-4.5,0.65)--  (-4,0.75);
   \draw[ultra thick,color=blue,dashed] (-4.5,0.65)--  (-5,0.55);
   \draw[ultra thick,color=blue,dashed] (-5.5,0.4)--  (-5,0.55);

 \draw (-0.5,2)-- (-1,4);
 \draw (0.5,2)-- (1,4);
 \draw (-2,2.05)-- (-1.5,4);
 \draw (-2.5,2.1)-- (-3,3.85);
 \draw (-3.5,2.2)-- (-4,3.75);
 \draw (-5.5,2.6)-- (-4.5,3.65);
 \draw (2,2.1)-- (1.5,4);
 \draw (2.5,2.1)-- (2.5,3.9);
 \draw (3,2.15)-- (3.5,3.8);
 \draw (4,2.25)-- (4,3.75);
 \draw (5,2.45)-- (4.5,3.65);
 \draw (5.5,2.6)-- (5,3.55);

 \draw (0.5,2-3)-- (1,4-3);
 \draw (0,2-3)-- (-0.5,4-3);
 \draw (2,2.05-3)-- (1.5,4-3);
 \draw (2.5,2.1-3)-- (3,3.85-3);
 \draw (3.5,2.2-3)-- (4,3.75-3);
 \draw (5.5,2.6-3)-- (4.5,3.65-3);
 \draw (-2,2.1-3)-- (-1.5,4-3);
 \draw (-2.5,2.1-3)-- (-2.5,3.9-3);
 \draw (-3,2.15-3)-- (-3.5,3.8-3);
 \draw (-4,2.25-3)-- (-4,3.75-3);
 \draw (-5,2.45-3)-- (-4.5,3.65-3);
 \draw (-5.5,2.6-3)-- (-5,3.55-3);

\node at (0.5,2) [circle,draw=black,fill=black,scale=0.5]{};
\node at (0,2) [circle,draw=black,fill=black,scale=0.5]{};
\node at (-0.5,2) [circle,draw=black,fill=black,scale=0.5]{};
\node at (-1,2) [circle,draw=black,fill=black,scale=0.5]{};
\node at (-1.5,2) [circle,draw=black,fill=black,scale=0.5]{};
\node at (-2,2.05) [circle,draw=black,fill=black,scale=0.5]{};
\node at (-2.5,2.1) [regular polygon,regular polygon sides=3,draw=black,fill=green,scale=0.4]{};
\node at (-3,2.15) [circle,draw=black,fill=black,scale=0.5]{};
\node at (-3.5,2.2) [circle,draw=black,fill=black,scale=0.5]{};
\node at (-4,2.25) [circle,draw=black,fill=black,scale=0.5]{};
\node at (-4.5,2.35) [circle,draw=black,fill=black,scale=0.5]{};
\node at (-5,2.45) [circle,draw=black,fill=black,scale=0.5]{};
\node at (-5.5,2.6) [circle,draw=black,fill=black,scale=0.5]{};
\node at (-6,3) [circle,draw=black,fill=black,scale=0.5]{};
\node at (1,2) [circle,draw=black,fill=black,scale=0.5]{};
\node at (1.5,2) [circle,draw=black,fill=black,scale=0.5]{};
\node at (2,2.05) [circle,draw=black,fill=black,scale=0.5]{};
\node at (2.5,2.1) [circle,draw=black,fill=black,scale=0.5]{};
\node at (3,2.15) [circle,draw=black,fill=black,scale=0.5]{};
\node at (3.5,2.2) [circle,draw=black,fill=black,scale=0.5]{};
\node at (4,2.25) [circle,draw=black,fill=black,scale=0.5]{};
\node at (4.5,2.35) [circle,draw=black,fill=black,scale=0.5]{};
\node at (5,2.45) [circle,draw=black,fill=black,scale=0.5]{};
\node at (5.5,2.6) [circle,draw=black,fill=black,scale=0.5]{};
\node at (6,3) [circle,draw=black,fill=black,scale=0.5]{};

\node at (0.5,4) [circle,draw=black,fill=black,scale=0.5]{};
\node at (0,4) [circle,draw=black,fill=black,scale=0.5]{};
\node at (-0.5,4) [circle,draw=black,fill=black,scale=0.5]{};
\node at (-1,4) [circle,draw=black,fill=black,scale=0.5]{};
\node at (-1.5,4) [circle,draw=black,fill=black,scale=0.5]{};
\node at (-2,3.95) [circle,draw=black,fill=black,scale=0.5]{};
\node at (-2.5,3.9) [circle,draw=black,fill=black,scale=0.5]{};
\node at (-3,3.85) [circle,draw=black,fill=black,scale=0.5]{};
\node at (-3.5,3.8) [circle,draw=black,fill=black,scale=0.5]{};
\node at (-4,3.75) [circle,draw=black,fill=black,scale=0.5]{};
\node at (-4.5,3.65) [circle,draw=black,fill=black,scale=0.5]{};
\node at (-5,3.55) [circle,draw=black,fill=black,scale=0.5]{};
\node at (-5.5,3.4) [circle,draw=black,fill=black,scale=0.5]{};
\node at (1,4) [circle,draw=black,fill=black,scale=0.5]{};
\node at (1.5,4) [circle,draw=black,fill=black,scale=0.5]{};
\node at (2,3.95) [circle,draw=black,fill=black,scale=0.5]{};
\node at (2.5,3.9) [circle,draw=black,fill=black,scale=0.5]{};
\node at (3,3.85)  [rectangle,draw=black,fill=red,scale=0.7]{};
\node at (3.5,3.8) [circle,draw=black,fill=black,scale=0.5]{};
\node at (4,3.75) [circle,draw=black,fill=black,scale=0.5]{};
\node at (4.5,3.65) [circle,draw=black,fill=black,scale=0.5]{};
\node at (5,3.55) [circle,draw=black,fill=black,scale=0.5]{};
\node at (5.5,3.4) [circle,draw=black,fill=black,scale=0.5]{};

\node at (0.5,2-3) [circle,draw=black,fill=black,scale=0.5]{};
\node at (0,2-3) [circle,draw=black,fill=black,scale=0.5]{};
\node at (-0.5,2-3) [circle,draw=black,fill=black,scale=0.5]{};
\node at (-1,2-3) [circle,draw=black,fill=black,scale=0.5]{};
\node at (-1.5,2-3) [circle,draw=black,fill=black,scale=0.5]{};
\node at (-2,2.05-3) [circle,draw=black,fill=black,scale=0.5]{};
\node at (-2.5,2.1-3) [circle,draw=black,fill=black,scale=0.5]{};
\node at (-3,2.15-3) [circle,draw=black,fill=black,scale=0.5]{};
\node at (-3.5,2.2-3) [circle,draw=black,fill=black,scale=0.5]{};
\node at (-4,2.25-3) [circle,draw=black,fill=black,scale=0.5]{};
\node at (-4.5,2.35-3) [circle,draw=black,fill=black,scale=0.5]{};
\node at (-5,2.45-3) [circle,draw=black,fill=black,scale=0.5]{};
\node at (-5.5,2.6-3) [circle,draw=black,fill=black,scale=0.5]{};
\node at (-6,3-3) [circle,draw=black,fill=black,scale=0.5]{};
\node at (1,2-3) [circle,draw=black,fill=black,scale=0.5]{};
\node at (1.5,2-3) [circle,draw=black,fill=black,scale=0.5]{};
\node at (2,2.05-3) [circle,draw=black,fill=black,scale=0.5]{};
\node at (2.5,2.1-3) [circle,draw=black,fill=black,scale=0.5]{};
\node at (3,2.15-3) [circle,draw=black,fill=black,scale=0.5]{};
\node at (3.5,2.2-3) [circle,draw=black,fill=black,scale=0.5]{};
\node at (4,2.25-3) [circle,draw=black,fill=black,scale=0.5]{};
\node at (4.5,2.35-3) [circle,draw=black,fill=black,scale=0.5]{};
\node at (5,2.45-3) [circle,draw=black,fill=black,scale=0.5]{};
\node at (5.5,2.6-3) [circle,draw=black,fill=black,scale=0.5]{};
\node at (6,3-3) [rectangle,draw=black,fill=red,scale=0.7]{};

\node at (0.5,4-3) [circle,draw=black,fill=black,scale=0.5]{};
\node at (0,4-3) [circle,draw=black,fill=black,scale=0.5]{};
\node at (-0.5,4-3) [circle,draw=black,fill=black,scale=0.5]{};
\node at (-1,4-3) [circle,draw=black,fill=black,scale=0.5]{};
\node at (-1.5,4-3) [circle,draw=black,fill=black,scale=0.5]{};
\node at (-2,3.95-3) [circle,draw=black,fill=black,scale=0.5]{};
\node at (-2.5,3.9-3) [regular polygon,regular polygon sides=3,draw=black,fill=green,scale=0.4]{};
\node at (-3,3.85-3) [circle,draw=black,fill=black,scale=0.5]{};
\node at (-3.5,3.8-3) [circle,draw=black,fill=black,scale=0.5]{};
\node at (-4,3.75-3) [circle,draw=black,fill=black,scale=0.5]{};
\node at (-4.5,3.65-3) [circle,draw=black,fill=black,scale=0.5]{};
\node at (-5,3.55-3) [circle,draw=black,fill=black,scale=0.5]{};
\node at (-5.5,3.4-3) [circle,draw=black,fill=black,scale=0.5]{};
\node at (1,4-3) [circle,draw=black,fill=black,scale=0.5]{};
\node at (1.5,4-3) [circle,draw=black,fill=black,scale=0.5]{};
\node at (2,3.95-3) [circle,draw=black,fill=black,scale=0.5]{};
\node at (2.5,3.9-3) [circle,draw=black,fill=black,scale=0.5]{};
\node at (3,3.85-3) [circle,draw=black,fill=black,scale=0.5]{};
\node at (3.5,3.8-3) [circle,draw=black,fill=black,scale=0.5]{};
\node at (4,3.75-3) [circle,draw=black,fill=black,scale=0.5]{};
\node at (4.5,3.65-3) [circle,draw=black,fill=black,scale=0.5]{};
\node at (5,3.55-3) [circle,draw=black,fill=black,scale=0.5]{};
\node at (5.5,3.4-3) [circle,draw=black,fill=black,scale=0.5]{};

\node at (5.5,-1.1)[circle,draw=black,fill=black,scale=0.5]{};
\node at (2.5,4.4) [circle,draw=black,fill=black,scale=0.5]{};

\node at (-5.1,3.9) {$x'_i$};
\node at (0,3) {$B_i$};
\node at (5.5,3.7) {$y'_i$};
\node at (-5.5,0.7) {$x'_{i+1}$};
\node at (0.3,0) {$B_{i+1}$};
\node at (6.5,0) {$y'_{i+1}$};

\node at (5.5,-1.4) {\Huge{$\ldots$}};
\node at (2.5,4.6) {\Huge{$\ldots$}};

\end{tikzpicture}
\caption{\label{figadd44} Example of two blocks $B_i$ and $B_{i+1}$ and the paths considered in the proof of Theorem \ref{thm path structure} in these  blocks (thick line: path  $P$; dashed lines: paths $Q_{2i-1}$, $Q_{2i}$, $Q_{2i+1}$ and $Q_{2i+2}$).}
\end{center}
\end{figure}
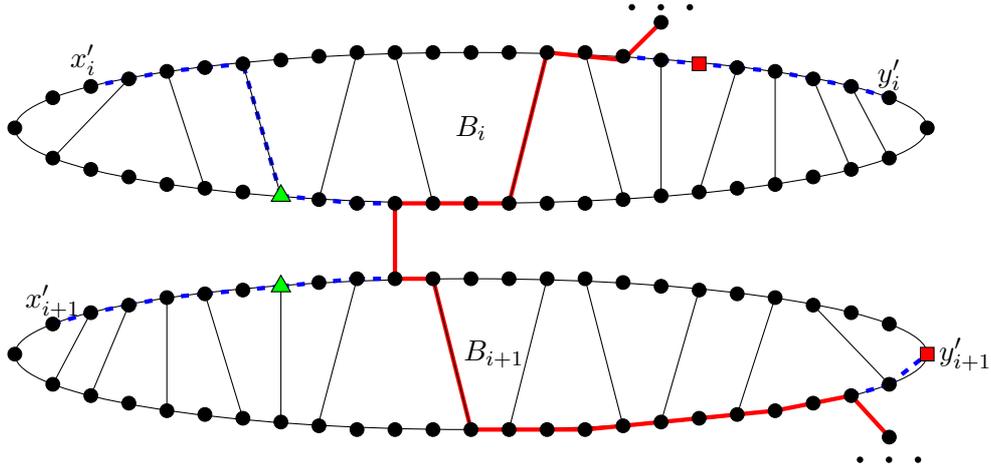

\medskip

Now we colour the paths $Q_j$, $j=1,2,\dots, \, 2k$. The distance from $P$ to any vertex of $Q_j$ in $G$ is the same as on $Q_j$, hence each $Q_j$ is a shortest path in $G$ between $P$ and the relevant vertex $x'_i$ or $y'_i$, respectively. Thus, by Lemma \ref{out2}, we can colour the vertices of each path $Q_{2i-1}$ ($i=1,2,\dots, k$) with Pattern~\eqref{pat5-15}, starting at the vertex of $Q_{2i-1}$ at distance two from $P$ (square vertices in Fig. \ref{figadd44}). Analogously we can colour the vertices of the paths $Q_{2i}$ ($i=1,2,\dots, k$) using Pattern~\eqref{pat5-15}, starting at the vertex of $Q_{2i}$ at distance three from $P$ (triangle vertices in Fig. \ref{figadd44}). 
Then the distance between vertices on distinct paths $Q_m$, $Q_n$ coloured with colour $5$ is at least $3+1+2$, the distance between vertices on distinct paths $Q_m$, $Q_n$ coloured with colour $6$ is at least $4+1+3$, etc. Therefore the defined colouring of the paths $Q_j$, $j=1,\dots, 2k$ satisfies the distance constraints of a packing colouring.

\medskip

Now we colour the remaining vertices of $G$. We start with colouring of the path $P$ with a pattern using colours $16,17,\dots, 50$ by Proposition \ref{preliminaries}.ii). Then we colour all uncoloured vertices of $Q_j$ ($j=1,\dots, 2k$) at distance one from $P$ with colours $51, 52, \dots, 152$ by Proposition \ref{preliminaries}.iii).
 For the remaining vertices of $Q_{2i}$ at distance two from $P$, the distance between any such vertices on $Q_{2m}$ and $Q_{2n}$ ($m,n\in \{ 1,2,\dots, k-1 \}$, $m\neq n$) is at least $2\vert m-n\vert +4$.  Hence we can colour these vertices with colours $153, \dots, 305$ by Proposition \ref{pathgap}. 
\end{proof}

\section{Concluding remarks}

In the previous sections, we have determined some classes of outerplanar graphs with finite packing chromatic number. As for lower bounds, we are (only) able to state the following, where symbols $\square$ and $\boxtimes$ denote the Cartesian and strong product of graph, respectively (see~\cite{kla}).

\begin{proposition}
There exists an infinite family of $2$-connected subcubic outerplanar graphs without internal faces and with packing chromatic number $5$.
\end{proposition}

\begin{proof}
It has been proven in \cite{God} that $\chi_{\rho}(G)=5$ for $G=P_n\square P_2$ and $n\ge6$.
\end{proof}

\begin{remark}
The graph $G$ illustrated in Figure \ref{lowerbound} is a $2$-connected subcubic outerplanar graph with packing chromatic number $7$. We verified by computer that every proper colouring of $G$ with $6$ colours is not a packing colouring and we found a packing colouring of $G$ with $7$ colours.
\end{remark}

\begin{figure}[t]

\begin{center}

\begin{tikzpicture}[scale=0.6]
\draw  (-2,0) -- (2,0);
\draw  (-2,0) -- (-1,2);
\draw  (2,0) -- (1,2);
\draw  (-1,2) -- (1,2);
\draw (-1,0) -- (-1.5,1);
\draw (1,0) -- (1.5,1);
\draw (-1,2) -- (-2,3);
\draw (-2,5) -- (-2,3);
\draw (-2,5) -- (-1,6);
\draw (1,6) -- (-1,6);
\draw (1,6) -- (2,5);
\draw (2,3) -- (2,5);
\draw (2,3) -- (1,2);
\draw (-2,3) -- (-4,2);
\draw (-2,5) -- (-4,6);
\draw (-4,2) -- (-4,6);
\draw (-4,3) -- (-3,2.5);
\draw (-4,5) -- (-3,5.5);
\draw (2,3) -- (4,2);
\draw (2,5) -- (4,6);
\draw (4,2) -- (4,6);
\draw (4,3) -- (3,2.5);
\draw (4,5) -- (3,5.5);
\draw (1,6) -- (2,8);
\draw (-1,6) -- (-2,8);
\draw (-2,8) -- (2,8);
\draw (-1,8) -- (-1.5,7);
\draw (1,8) -- (1.5,7);

\node at (2,0) [circle,draw=black,fill=black,scale=0.5]{};
\node at (1.5,1) [circle,draw=black,fill=black,scale=0.5]{};
\node at (1,2) [circle,draw=black,fill=black,scale=0.5]{};
\node at (-2,0) [circle,draw=black,fill=black,scale=0.5]{};
\node at (-1,0) [circle,draw=black,fill=black,scale=0.5]{};
\node at (1,0) [circle,draw=black,fill=black,scale=0.5]{};
\node at (-1.5,1) [circle,draw=black,fill=black,scale=0.5]{};
\node at (-1,2) [circle,draw=black,fill=black,scale=0.5]{};
\node at (2,8) [circle,draw=black,fill=black,scale=0.5]{};
\node at (1.5,7) [circle,draw=black,fill=black,scale=0.5]{};
\node at (1,6) [circle,draw=black,fill=black,scale=0.5]{};
\node at (-2,8) [circle,draw=black,fill=black,scale=0.5]{};
\node at (-1,8) [circle,draw=black,fill=black,scale=0.5]{};
\node at (1,8) [circle,draw=black,fill=black,scale=0.5]{};
\node at (-1.5,7) [circle,draw=black,fill=black,scale=0.5]{};
\node at (-1,6) [circle,draw=black,fill=black,scale=0.5]{};
\node at (2,3) [circle,draw=black,fill=black,scale=0.5]{};
\node at (2,5) [circle,draw=black,fill=black,scale=0.5]{};
\node at (4,2) [circle,draw=black,fill=black,scale=0.5]{};
\node at (4,6) [circle,draw=black,fill=black,scale=0.5]{};
\node at (4,3) [circle,draw=black,fill=black,scale=0.5]{};
\node at (4,5) [circle,draw=black,fill=black,scale=0.5]{};
\node at (3,5.5) [circle,draw=black,fill=black,scale=0.5]{};
\node at (3,2.5) [circle,draw=black,fill=black,scale=0.5]{};
\node at (-2,3) [circle,draw=black,fill=black,scale=0.5]{};
\node at (-2,5) [circle,draw=black,fill=black,scale=0.5]{};
\node at (-4,2) [circle,draw=black,fill=black,scale=0.5]{};
\node at (-4,6) [circle,draw=black,fill=black,scale=0.5]{};
\node at (-4,3) [circle,draw=black,fill=black,scale=0.5]{};
\node at (-4,5) [circle,draw=black,fill=black,scale=0.5]{};
\node at (-3,5.5) [circle,draw=black,fill=black,scale=0.5]{};
\node at (-3,2.5) [circle,draw=black,fill=black,scale=0.5]{};
\end{tikzpicture}
\end{center}
\caption{An outerplanar subcubic graph with packing chromatic number 7.}
\label{lowerbound}
\end{figure}
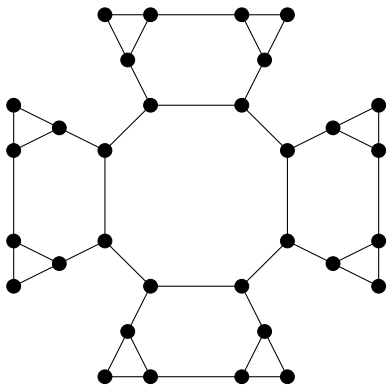

Bre{\v{s}}ar et al. \cite{BrePa} have proven that for any finite graph $G$, the graph $G\boxtimes P_{\infty}$ has finite packing chromatic number. The degree of $G\boxtimes P_{\infty}$ can be arbitrary large. This property illustrates the fact that the  degree of a graph is not the only parameter to consider in order to have finite packing chromatic number. Maybe the fact that the weak dual is a path (and is not any tree) helps to bound the packing chromatic number. It remains an open question to determine if the packing chromatic number of subcubic outerplanar graphs is finite or not.

\section*{Acknowledgments} 
The authors thank the referees for their judicious comments and Mahmoud Omidvar for his precious help, in particular for providing Pattern~\eqref{pat5-15} of proof of Theorem \ref{out1}.
First author was partly supported by the Burgundy Council under grant \#CRB2011-9201AAO048S05587. Second author was partly supported by project P202/12/G061 of the Czech Science Foundation.

\end{document}